\documentclass[11pt]{article}
\usepackage{amssymb,amsmath,amsthm,enumerate,comment,mathtools,thmtools}
\usepackage[margin=1in]{geometry}
\usepackage{xcolor}
 \usepackage[
   pagebackref,
   colorlinks=true,
   urlcolor=black,
   linkcolor=black,
   citecolor=black,
 ]{hyperref}
\usepackage[nameinlink]{cleveref}
\usepackage[utf8]{inputenc}
\usepackage{graphicx}
\usepackage{xcolor}
\usepackage{bbm}
\usepackage{physics}
\usepackage{etoolbox}

\newcommand{\mbe}[1]{\textcolor{purple}{{#1}}}

\def\E{\mathop{{}\mathbb{E}}}

\newcommand{\eps}{\varepsilon}

\newcommand{\ind}{\mathbbm{1}}

\DeclareMathOperator*{\argmax}{arg\,max}

\newcommand{\Prob}{\mathbb{P}}
\newcommand{\N}{\mathcal{N}}

\newcommand{\dist}{F}

\newcommand{\rev}{\mathsf{Rev}}

\newcommand{\ub}{\mathsf{ub}}
\newcommand{\us}{\mathsf{us}}
\newcommand{\diff}{\mathsf{diff}}

\newcommand{\Prange}[3]{
 \ifx #3\infty
 \Prob[#2 \leq v_{#1}]
\else \Prob[#2 \leq v_{#1} < #3]
\fi
}

\newcommand{\truncV}[2]{
V_{#1}(#2)
}

\newcommand{\items}{{\mathcal{M}}}
\newcommand{\rem}{\mathsf{rem}}

\DeclarePairedDelimiter{\set}{\{}{\}}

\DeclarePairedDelimiter{\sq}{[}{]}
\DeclarePairedDelimiter{\paren}{\lparen}{\rparen}
\DeclarePairedDelimiterX{\cond}[2]{[}{]}{#1\,\delimsize\vert\,\mathopen{} #2}

\newtheorem{theorem}{Theorem}[section]
\newtheorem{lemma}[theorem]{Lemma}

\newtheorem{claim}[theorem]{Claim}

\newtheorem{proposition}[theorem]{Proposition}

\newtheorem{definition}[theorem]{Definition}

\title{Bundling in Oligopoly:\\ Revenue Maximization with Single-Item Competitors}
\author{
Moshe Babaioff \thanks{\href{mailto:moshe.babaioff@mail.huji.ac.il}{moshe.babaioff@mail.huji.ac.il}.  Moshe Babaioff's research is supported in part by a Golda Meir Fellowship.}\\
{\small Hebrew University of Jerusalem}
\and
Linda Cai\thanks{\href{mailto:tcai@princeton.edu}{tcai@princeton.edu}.}\\
{\small Princeton University}
\and
Brendan Lucier\thanks{\href{mailto:brlucier@microsoft.com}{brlucier@microsoft.com}.}\\
{\small Microsoft Research} 
}
\date{}

\begin{document}

\maketitle

\begin{abstract}
    We consider a principal seller with $m$ heterogeneous products to sell to an additive buyer {over independent items}.  The principal can offer an arbitrary menu of product bundles, but faces competition from smaller and more agile single-item sellers.  
    The single-item sellers choose their prices after the principal commits to a menu, potentially under-cutting the principal's offerings.  We explore to what extent the principal can leverage the ability to bundle product together to extract revenue.

    Any choice of menu by the principal induces an oligopoly pricing game between the single-item sellers, which may have multiple equilibria.  When there is only a single item this model reduces to Bertrand competition, for which the principal's revenue is $0$ at any equilibrium, so we assume that no single item's value is too dominant. We establish an upper bound on the principal's optimal revenue at every 
    equilibrium: the expected welfare after truncating each item's value to its revenue-maximizing price. 
    Under a technical condition on the value distributions --- that the monopolist's revenue is sufficiently sensitive to price --- 
    we show that the principal seller can simply price the grand-bundle and ensure (in any equilibrium) a constant approximation to this bound (and hence to the optimal revenue).
    We also show that for some value distributions violating our conditions, grand-bundle pricing does not yield a constant 
    approximation to the optimal revenue in any equilibrium. 
\end{abstract}

\section{Introduction}

True monopolies are rare.  Even a large, dominant seller in a given market will almost certainly need to contend with small, agile competitors who attempt to offer similar products.  For example, suppose that a large firm has many heterogeneous goods to sell to a buyer who values them additively and independently.  This setting has been studied extensively in the algorithmic mechanism design literature when the seller is a monopolist, and it is known that the seller can often extract more revenue (sometimes close to the full social welfare) by bundling goods together and selling them as packages {instead of selling each item separately} ~\cite{HartN12,HartN17,LiY13,BabaioffILW14}.
But this sort of aggressive bundling strategy carries risk: if a smaller competitor were able to replicate one of the products, they could sell it as a stand-alone offering and attract away customers who have an especially strong preference for it relative to the others.  {On the other hand, selling items \`{a} la carte leaves the firm vulnerable to being undercut on prices.  Either way, the threat of competition even from smaller sellers of individual products can substantially impact the firm's sales strategy and revenue.}  

Our goal in this paper is to understand the ability of a multi-product principal seller to extract revenue 
in settings where he faces competition from agile item sellers.  
In our model there is a single multi-good seller (the principal) who has $m$ items to sell, plus a pool of item sellers who each has only a single one of the $m$ item types to sell.  For most of the paper we will assume that there is exactly one competing item seller per item type.\footnote{Our results extend naturally to having an arbitrary and possibly different number of item sellers for each good (including none at all), as we discuss in our model extensions below.}
The principal, acting as a market leader, first selects and commits to an arbitrary menu of product bundles.  Then, simultaneously, each item seller picks her own deterministic menu for selling her item (equivalently, picks a price for her item). The item sellers can randomize over different prices. The buyer, aiming to maximize her utility, then decides on a subset to buy from the principal seller, and on the set of items she buys from the item sellers.

The presence of the item sellers impedes the ability of the principal to extract revenue because of their ability to undercut the principal's prices.  In the special case of only a single item, this is classic Bertrand competition.  Indeed, if the principal chooses any price $p > 0$, a competing item seller would simply undercut $p$ by some small amount and steal away all of the principal's sales.  Thus, in the case of a single item type, it is impossible for the principal to generate any revenue.  However, if there are two or more items, the principal has a power that the item sellers do not: the power to bundle items together and sell them as a package.  We ask: 
{\emph{What are the limits on the revenue of the principal due to the competition? Under which market conditions can the principal extract significant revenue even in the presence of such competition?}}

Any given choice of menu by the principal seller defines a downstream oligopolistic pricing game played by the item sellers.  Our solution concept in that downstream game is mixed 
Nash equilibrium (NE): each item seller selects a distribution over price choices that maximizes her expected revenue given the {menu of the principal and the distributions selected by}  
the other item sellers.  
{With a discrete pricing space} a mixed NE is guaranteed to exist, but is not unique in general.  
{We are able to prove an upper bound on the revenue attainable by the principal that is robust to the choice of equilibrium: it holds for any choice of menu and any equilibrium of play by the item sellers.} 
{We show that} 
the principal's revenue cannot be greater than the expected sum of truncated item values, where each item's value is truncated to its (maximal) revenue-maximizing price.

\begin{theorem}
\label{thm:intro.bound} Suppose the buyer has independent item values $v_i \sim F_i$ and that $r_i \in \arg\max_r\{ r(1-F_i(r)) \}$ is the maximal revenue-maximizing price for item $i$.  Then for any menu of the principal seller and any mixed Nash equilibrium among the item sellers, the expected revenue of the principal seller is at most the buyer's expected truncated social welfare $\E_{\vec{v}{\sim \times_i \dist_i}}[\sum_i \min\{v_i, r_i\}]$.
\end{theorem}

We note that this upper bound may be substantially lower than the revenue obtainable by a monopolist.  For example, suppose the item value distributions are i.i.d.\ draws from an equal revenue distribution supported on $[1,H]$ {for some large $H>m$}
,\footnote{This is the distribution $F$ over $[1,H]$ such that $1-F(z) = 1/z$ for all $1 \leq z < H$.} skewed slightly so that $1$ is the unique revenue-maximizing price.  In this instance it is known that a monopolist seller can obtain revenue $\Omega(m \log m)$ by appropriately pricing the grand bundle of all items.  But  the expected truncated welfare, our upper bound on the principal's revenue in our competitive setting, is {much lower, just} $O(m)$.

What makes it harder for the principal to extract revenue than a monopolist?  We recall that for a monopolist, it is possible to achieve a constant fraction of the optimal revenue with a simple menu that either prices the grand bundle of all items, or sets a separate price on each individual item and allows the buyer to purchase \`{a} la carte \cite{BabaioffILW14}.  In our competitive setting, however, the latter is not an option: setting a separate price on each item reduces the problem to $m$ independent instances of Bertrand competition, leading to each item's price being undercut {by its item seller} and no revenue being generated {for the principal}.  On the other hand, by bundling items together, the principal might encourage item sellers to keep their prices high; however, as we show, an item seller will never choose a price higher than their maximal monopolist reserve price $r_i$.  This ultimately prevents the principal from extracting revenue from higher values, leading to our upper bound on the principal's optimal revenue.

We show that, under certain assumptions on the value distributions, it is possible for the principal to guarantee a constant approximation to 
the expected truncated social welfare (which, by Theorem~\ref{thm:intro.bound}, upper bounds the optimal revenue).
Moreover, this is achievable with a simple menu that offers only the grand bundle of all products at a carefully-selected price.

\begin{theorem}[Informal]
\label{thm:intro.main}
{Suppose each item's value distribution is $(\lambda, C)$-price-sensitive (described below) and 
{that the variance of the truncated social welfare is sufficiently high (as a function of $\lambda$, $C$, and $\max_i r_i$).}
Then there exists a price $p$ at which the principal can sell the grand bundle of all items such that, at any mixed NE for the item sellers, the principal's revenue is at least $1/3$ of the expected truncated social welfare.} 
\end{theorem}

The assumptions we impose on the value distributions are of two types.  First, we require that the total variance of the  sum of {truncated} values is high relative to each individual item's contribution.  This implies that no individual item contributes too much to the expected {truncated} welfare, or has too high an influence on the likelihood of sale to the principal.  This is an expansion of the requirement that there be more than one item for sale (which, recall, is necessary for the principal to obtain any revenue at all).  For example, if we have multiple items but a single item contributes almost all of the value, then we are ``essentially'' in the single-item case and similar issues persist.  {To give another example of what can go wrong if individual item sellers have high influence, suppose each of the $m$ items has value 1 for sure.  Then for any price $p$ that the principal sets on the grand bundle, there is a Nash equilibrium in which every item seller sets price $p/m$, precisely coordinating to undercut the grand bundle price and resulting in revenue $0$ for the principal.  In this case, even though no individual item dominates the welfare, each individual item seller still has substantial influence on the buyer's aggregate purchasing decisions at equilibrium.  Our condition that the variance is sufficiently high excludes such a scenario, since no individual item price will be significantly predictive of the buyer's utility-maximizing behavior.}  

Second, we require that each item's value distribution is non-trivially ``price sensitive", by which we mean that the revenue obtained decreases sufficiently (parameterized by $\lambda$) at prices sufficiently below the revenue-maximizing price (parameterized by $C$). This is a technical condition that essentially rules out point-mass distributions and equal-revenue distributions, requiring that the distributions are sufficiently far from these edge case distributions (and the bound we require on the number of items is parameterized by how far they are).

We leave a formal definition of the price sensitivity condition to later sections (see Section~\ref{sec:approxopt}).  Instead, we now note some simpler sufficient conditions that imply our approximation result: 
\begin{enumerate}
    \item In the i.i.d.\ case where all item values are drawn from a distribution $G$, if $G$ has a unique revenue-maximizing price $r^*$ that is greater than the minimum positive value in the support of $G$, our conditions will be satisfied for all sufficiently large $m$.
    \item For every $i$, there exists a $\delta$ such that if the value distribution of item $i$
    has revenue-optimizing price $r_i \geq r_{min} > 0$, density at least $\delta/r_i > 0$ on $[0, r_i]$, 
    and strictly concave revenue curves (in price space) with second derivatives at most $-\delta/r_i < 0$, 
    our conditions will be satisfied for all sufficiently large $m$. 
\end{enumerate}

It is worth noting that while in Theorem~\ref{thm:intro.main} and its corollaries we describe our assumptions as though they must hold for all items, this is not strictly necessary.  Our results degrade gracefully in that sense that if some subset $S$ of the items satisfy the necessary conditions, then the principal's revenue can approximate the expected truncated welfare of the items in $S$.  Indeed, the principal can always choose to ignore any items not in $S$, 
{leaving the corresponding item sellers to act as monopolists (so in every equilibrium they price at their revenue-maximizing prices)}. 

While 
the assumptions we impose are technical and we suspect they can be relaxed, we show that some form of assumption on the 
value distributions is necessary for the grand bundle to obtain a good revenue-approximation result.  
{
We construct a family of problem instances for 
which, unlike in Bertrand competition, the principal's optimal menu generates substantial revenue.} 
{For this family, the revenue of the best partition menu (in which the items are partitioned into disjoint bundles and each is assigned a price) is higher than the revenue of any grand-bundle pricing by an arbitrarily large factor ($\Theta(m)$).} 
We leave open the question of how to construct an approximately optimal menu (which must necessarily be more complex than pricing the grand bundle) for general distributions.  In particular, is there always an approximately revenue-optimal menu for the principal that is a partition menu?  If so, is it possible to construct such an approximately-optimal menu in polynomial time?

\medskip
\noindent
\textbf{Proof Techniques.}
We now briefly describe the proof of Theorem~\ref{thm:intro.main} and how it differs from corresponding results for a monopolist seller. 
For a monopolist, if the buyer's expected total value from all items is sufficiently concentrated around its mean, the principal can collect {a constant fraction}
of this surplus by pricing the grand bundle {at a price that is constant fraction of the mean.}  
The idea behind Theorem~\ref{thm:intro.main} is similar, but the principal must additionally consider the impact of prices chosen by the item sellers in equilibrium.  If the buyer's expected \emph{truncated} welfare is sufficiently concentrated around its mean, the principal can attempt to extract {a constant fraction of} this welfare as revenue with
{a price that is constant fraction of the mean.
This can potentially work if the competing item sellers set prices close to their optimal monopolist reserves $r_i$, but (unlike for a monopolist) this approach would certainly}
fail if the item sellers set prices so low that the buyer rarely buys from the principal.  To show that this doesn't happen
{under the assumptions of Theorem~\ref{thm:intro.main},} we use the Berry-Esseen theorem to argue that 
the expected truncated welfare is sufficiently \emph{anti}-concentrated (in addition to our earlier requirement of being concentrated enough).
This variability implies that either the principal seller sells the grand bundle with constant probability (in which case we are already done, {since the price is a constant fraction of the benchmark}), 
or each individual item seller has limited influence on whether or not the buyer purchases from the principal.  In the latter case, item sellers behave like (approximate) monopolists, so as long as their revenue is non-trivially sensitive to their choice of price, they will choose prices close {enough} to their revenue-maximizing prices at equilibrium {(at least $C\cdot r_i$ for every $i$ where $C>0$ is a constant that depends on the value distribution)}.  
But now, since the item sellers are setting high prices, the principal effectively acts as a monopolist with respect to the truncated value distributions and can capture a constant fraction of the truncated welfare by its choice of price on the grand bundle.

\medskip
\noindent
\textbf{Model extensions.} Our model focuses on the case that there is one item seller for each item the principal seller has. Our results can easily be extended to the case that there is any number of item sellers for each good (including none at all). In any case that  there are several item sellers which all supply exactly the same item, the price of that item will drop to 0 {in every equilibrium} (standard Bertrand competition between these item sellers), so such items cannot contribute to the revenue of the principal at all.\footnote{The principal can extract the exact same revenue with or without these items, since it can be assumed the buyer would purchase those items from the item sellers at price $0$ regardless, so would behave as though the items are not offered by the principal.} Let us now move to consider items that have no item seller at all. The principal seller is a monopolist on these items. Let $X$ be the set of items or which the principal is a monopolist, and $Y$ be the set of items with exactly one item seller for each item. Similarly to the claim proved by Hart and Nisan\cite{HartN12}, we observe that the revenue of principal in this model is bounded by the revenue obtainable from $X$ plus the welfare of $Y$. By the result of \cite{BabaioffILW14}, either selling the items in $X$ separately, or bundling all of them together, provides constant approximation to the optimal revenue of the monopolist. Combining the appropriate menu for $X$ with our offer of $Y$ as a bundle in the case that we are able to obtain a positive result (constant approximation to the expected truncated social welfare), we get constant approximation to the optimal menu over the set $X\cup Y$.

\medskip
\noindent
\textbf{Organization.}  In Section~\ref{sec:model} we formalize our model, including the market timing and equilibrium concept.  In Section~\ref{sec:upperbound} we establish our upper bound on the principal's optimal revenue, by first analyzing the structure of the buyer's purchasing decisions (Section~\ref{subsec:buyerDecision}) and the best-responses prices of the item sellers (Section~\ref{subsec:itemSeller}) then bounding the principal's 
revenue at equilibrium (Section~\ref{subsec:upperbound}).  In Section~\ref{sec:approxopt} we prove our main result: a constant revenue approximation for the principal under a general set of conditions, followed by some market structures that satisfy these conditions.
Finally, in Section~\ref{sec:brev-prev-lb} we show that the best achievable revenue from pricing the grand bundle can be an arbitrarily poor approximation to the principal's optimal revenue.

\subsection*{Additional Related Work}
\textbf{Auction Design for a Monopolist Seller.} 
Revenue optimal auction design is known to be a hard problem, even when the designer has monopoly over the market.
Specifically, when an auction involves multiple items, even when there is only one additive buyer, the revenue optimal auction of a monopolist may involve unrealistic features such as requiring randomization, lack of revenue monotonicity, and being computationally intractable \cite{Thanassoulis04, ManelliV07, Pavlov11, HartN13, DaskalakisDT14, HartR15, DaskalakisDT17}. 
{On the other hand, simple auction formats have been shown to provide a constant approximation to the monopolist optimal revenue for an additive buyer~\cite{HartN12,HartN17,LiY13,BabaioffILW14}.}
\cite{CaiDW16} establishes a duality framework which gives a tractable benchmark that bounds the optimal revenue within constant factor. The benchmark has subsequently been heavily utilized in the design and analysis of simple and approximately optimal auctions \cite{LiuP18, CaiZ17, EdenFFTW17b, EdenFFTW17a, BrustleCWZ17, DevanurW17, FuLLT18, BeyhaghiW19, CaiS21}. Specifically, in the setting of one additive buyer, the better between selling separately and selling the grand bundle is a constant approximation to the optimal revenue. 
Relative to that literature, we initiate a study of revenue benchmarks for a non-monopolistic principal who acts as a market leader and can design a menu (possibly with bundling) to which other sellers might respond.  

\textbf{Oligopoly. } Oligopoly has long been a topic studied in economics. Most literature studies price competition between two or several large firms who offer similar sets of products. See, for example, \cite{chen1997equilibrium,armstrong2010competitive} for examples of duopoly analysis. Particularly related to this paper is a line of work that studies how much a firm who has obtained monopolist status in one product can use bundling to deter single product competitors in another product, and its legislative implications 
\cite{whinston1989tying,nalebuff2004bundling,vickers2005abuse,avenali2013bundling}. 
More recent work has considered competition between multiple firms who can each bundle their products, such as \cite{khodabakhshian2020competitive} which considers the case of two firms each with two products.

In contrast, our work studies the strategies of a large firm that offers multiple products, given that small competitors already exist in the market for each product. Such situations is common for tech markets (e.g. security products, graphic products), where the fixed cost of maintaining a small business is relatively lower compared to other. \cite{mantovani2013strategic} studies this setting, but under the restricted assumption that the principal seller sells two products. \cite{shuai2022dominant} also studies the two product setting, but where the principal seller is competing with several small businesses for each product. Our paper initiates a systematic study of oligopoly with single item competitors, where we allow any number of items to be offered by the principal seller, and we focus on the principal seller's problem of constructing an approximately revenue-optimal menu.

\section{Model and Preliminaries} \label{sec:model}

We consider a market for a set $\items$ of $m$ item types. Items are indivisible. 
There is a single buyer in the market, interested in buying at most one item from each item type, and having a private additive {non-negative} valuation over subsets of items. We denote the value of item $i\in \items$ by {$v_i\geq 0$},  
and thus the value of a subset of items $B \subseteq \items$ is $\sum_{i \in B} v_i$. 

{We assume that the valuation of the buyer} 
is drawn from a 
prior distribution $\dist$, which is common knowledge among all participants, with the value of each item $i$ drawn independently from the distribution $\dist_i$ {with support in $[0,\infty)$}, 
so $\dist = \times_{i=1}^m \dist_i$. 
We assume (wlog) that $\dist_i$ {is non-trivial ($F_i[v_i=0]<1$). 
{Let $\rev_{F_i}(\cdot)$ be the revenue function for $F_i$: the revenue by pricing at price $x$ is $\rev_{F_i}(x)= x \cdot (1-\dist_i(x))$. \footnote{Where we write $\dist_i(x) = \Pr_{v_i \sim \dist_i}[v_i < x]$.} When clear from the context we use $\rev_{i}(\cdot)$ to denote $\rev_{\dist_i}(\cdot)$.
We assume that $\dist_i$ has a maximal revenue-maximizing (Myerson) price. 
That is, there exist a price $r_i$ such that $\rev_{F_i}(r_i)\geq \rev_{F_i}(x)$ for every price $x$, 
and $\rev_{F_i}(r_i)> \rev_{F_i}(y)$ for every price $y>r_i$.}  
} 
The realized value profile of the buyer  $\vec{v} = (v_1, \cdots, v_m)$ is drawn from  the prior product distribution $\dist = \times_{i=1}^m \dist_i$. Only the buyer knows the realized values.

In this market there are multiple sellers that supply the desired items, {all with $0$ cost of production}. 
There is one seller that we call the \emph{principal seller}
who can supply any subset of the $m$ items.
In addition to the principal seller, there are $m$ \emph{item sellers}, where item seller $i$ 
can supply a copy of item $i\in \items$ only (that is, an item that is completely identical to the $i$'th item supplied by the principal seller).  
The sellers have no value for their items, only value for money received from the buyer.
All agents are assumed to be risk neutral, with quasi-linear utility functions. Each seller aims to maximize their revenue obtained.  The buyer aims to maximize their expected utility, defined as the difference between their valuation for the acquired items and the total price paid. 

\paragraph{{Market Timing}}
Informally, our market operates as follows.  First, the principal seller chooses a deterministic menu of offers, and commits to it. Then, each item seller simultaneously picks her own deterministic menu for selling her item (equivalently, picks a price for her item). 
Finally, the buyer
selects a subset of items to buy from the principal seller and a subset of items to buy from the item sellers, either of which might be empty. 
More formally, the timing of our market 
is as follows:

\begin{enumerate}
 \item The principal seller first chooses and commits to a deterministic menu of prices, one for each subset of items $T \subseteq \items$.  We will write $p \colon 2^\items \to \mathbb{R}_{\geq 0}$ 
 to denote the price menu selected by the principal seller, so that $p(T)$ is the price assigned to subset $T \subseteq \items$.\footnote{As is standard, it will sometime be convenient to think about a partial menu that does not explicitly list all subsets as menu entries (e.g., only prices the grand bundle). In this case, we assume free disposal so that the price of a set $T$ is the cheapest price of all $S$ such that $T\subseteq S$ in the menu, and infinite if there is no such set $S$.
 } 
{
 The principal has to offer the option of buying nothing and paying zero ($p(\emptyset) = 0$). 
 }  
 {We denote by $\mathcal{P}$ the space of all such deterministic price menus. }
    \item Next, {knowing the menu $p$ picked by the principal seller,}
    the item sellers simultaneously select prices for their corresponding items.  Each item seller $i$ chooses a single non-negative price $q_i \geq 0$ at which to offer item $i$.
    \item Finally, {the buyer value profile $(v_1,\ldots, v_M)\sim \times_i \dist_i$ is realized, and} the buyer purchases items from the principal and/or some item sellers.  This involves choosing a {single} subset $T \subseteq \items$ of items from the menu offered by the principal {and paying $p(T)$ to the principal}, as well as choosing a subset $U \subseteq \items$ of items from individual item sellers, {paying $q_i$ to each item seller $i\in U$}. 
    The buyer acts as a price-taker and will always make purchase decisions to maximize their own utility.  In case of indifference, we assume that the buyer breaks ties in favor of minimizing the revenue of the principal seller,\footnote{{As item sellers price their items after the principal seller, they can always slightly reduce their price to win any tie.}} and then in favor of maximizing the number of items purchased from the item sellers.  
\end{enumerate}
Thus, for menu $p$ and item sellers prices $\vec{q}$,
the buyer's utility when her  values for the items are $\vec{v}$ is $\sum_{i \in T \cup U} v_i - p(T) - \sum_{i \in U}q_i$.  The principal seller obtains revenue $p(T)$, and each item seller $i$ obtains either revenue $q_i$ if $i \in U$, and $0$ otherwise. 

\paragraph{{The Game and Equilibrium Concepts}} 

Given a value distribution $\dist$ for the buyer, 
a menu $p\in \mathcal{P}$ picked by the principal seller together with the buyer behaviour as specified above, induce a simultaneous-move game 
{$G_{p,\dist}$} among the item sellers.
{A mixed strategy $s_i$ of item seller $i$ is a distribution over prices $q_i$.  {We will assume for technical convenience that item prices are constrained to be multiples of some arbitrarily and sufficiently 
small increment $\epsilon > 0$ (and that $q_i \leq 1$ for each $i$, which is without loss of generality since $v_i \leq 1$)}. 
Note that item sellers pick strategies in the game {$G_{p,\dist}$}, thus the strategy may depend on $p$. }
The space of mixed strategies of item seller $i$ is denoted by $S_i$. 
Given $p$ and a profile of mixed strategies $\vec{s}=(s_1,s_2,\ldots, s_m)$, we denote the (expected) utility of item seller ${i}$ in the game $G_{p,\dist}$ 
by $u_{i, \dist}({p}, \vec{s})$ (where the expectation is over the distribution of values $\vec{v}\sim \dist$ and distribution of item seller prices, $q_i\sim s_i$ for every $i$). 
When $\dist$ is clear from the context we sometimes use  $u_{i}({p}, \vec{s})$ to the utility $u_{i, \dist}({p}, \vec{s})$.

We are now ready to define equilibria of the game $G_{p,\dist}$.

\begin{definition}
Consider a market with value distribution $\dist=\times_i \dist_i$. 
For any 
{principal menu $p \in {\mathcal{P}}$}, a  
(mixed) 
\textbf{Nash Equilibrium} of the induced game 
{$G_{p,\dist}$} is a profile of 
strategies $(s_{1}^*, s_{2}^*, \dots, s_{m}^*) \in  \times_{i=1}^m S_{i}$ among item sellers  such that for every item seller $i$ and for all $s_{i}\in S_{i}$ 
we have:
$$
u_{i, \dist}( {p}, (s_{i}^*, s_{{-i}}^*)) \geq u_{i, \dist}( {p}, (s_{i}, s_{{-i}}^*)),
$$
where $s_{{-i}}^*$ represents the strategy profile of all item sellers other than item seller $i$.
\end{definition}

We first briefly comment on the existence of Nash equilibria in $G_{p, F}$. For any $\dist$ and $p$, consider a game $G'_{p, \dist}$, which is identical to $G_{p, \dist}$, except that the pricing space for each item seller $i$ is restricted to $[0, r_i]$. Since we assume prices are multiples of some arbitrarily small increment, the game $G'_{p, \dist}$ is finite and thus at least one mixed Nash equilibrium is guaranteed to exist. We will soon see in \Cref{sec:upperbound} (in particular \Cref{lem:sellerUtility2}) that for any item seller $i$, pricing above $r_i$ is a weakly dominated strategy, so any deviation to a price above $r_i$ is not beneficial. As a result, a mixed Nash equilibrium in $G'_{p, \dist}$ is also a mixed Nash equilibrium in $G_{p, \dist}$. 

Given $\dist=\times_i \dist_i$, a menu $p$ and a profile of mixed strategies $\vec{s}=(s_1,s_2,\ldots, s_m)$, the \emph{revenue (utility) of the principal seller from $(p,\vec{s})$} in the game $G_{p,\dist}$ is defined to be the expected payment to the principal with menu $p$ when $\vec{v}\sim \dist$, and $q_i\sim s_i$ for every item seller $i$ (given the buyer behaviour). 
The \emph{revenue of the principal seller from menu $p$} is defined to be the lowest revenue, over all Nash equilibrium profiles $\vec{s}$, of the principal seller from $(p,\vec{s})$ in the game $G_{p,\dist}$.
The principal seeks to find $p$ with high revenue from menu $p$.

\section{An Upper Bound on the Principal's Revenue} \label{sec:upperbound}
In this section we 
analyze
the behavior of the agents in our model. We will discuss both the buyer utility and the item sellers' utilities, and consequently, structures in their decision making processes. Using these observations, we will present a benchmark that upper bounds the {revenue the principal obtains by any menu she picks, in any mixed Nash equilibrium of the induced game between the item sellers.} {Missing proofs in this section can be found in \Cref{app:sec:upperbound}.}

\subsection{Understanding the Buyer's Optimal Decision} \label{subsec:buyerDecision}

We begin by considering the purchase decision of a buyer with value profile $\vec{v}$ for the items, given the pricing decisions $(p,\vec{q})$ of all sellers. In our analysis, we will represent the buyer's 
utility with different purchase options using the following notation: 

{
\textbf{Utility from the Principal}: 
{We use $\ub_{P, \vec{v}}(T, p)$ to denote the maximum utility the buyer can attain by purchasing any subset $S\subseteq T$ of items from the principal seller, when  the buyer's realized values are $\vec{v}$ and the principal's price menu $p$. That is:}
\[ \ub_{P, \vec{v}}(T, p) = \max_{S \subseteq T} \left( \sum_{i \in S}v_i - p(S)\right) . \]
}

\textbf{Utility from Item Sellers}: 
{
We use $\ub_{I, \vec{v}}(T, \vec{q})$ to denote the maximum utility the buyer can attain by purchasing any subset $S\subseteq T$ of items from the item sellers, when  the buyer's realized values are $\vec{v}$ and the item sellers price vector is $\vec{q}$. That is:}
\[ \ub_{I, \vec{v}}(T, \vec{q}) = \max_{S \subseteq T} \sum_{i \in S}(v_i - q_i). \]
{Clearly, items that are priced higher than their value will not be acquired.}

\textbf{Utility of Best Response Considering Seller Strategies}: {Fix some realized values $\vec{v}$ of the buyer, a price menu $p$ of the principal seller, and a vector of (possibly randomized) strategies $(s_1, \cdots, s_n)$ of the item sellers.
The buyer's expected utility from their utility-maximizing purchasing decision will be denoted by $\ub_{\vec{v}}(p, s_1, \cdots, s_n)$ and is computed as follows:
}
\begin{align}
\label{eq:buyer.max}
\ub_{\vec{v}}(p, s_1, \cdots, s_n) = \E_{\vec{q} \sim \vec{s}}\left[\max_{T \subseteq \items} \left\{\ub_{P,\vec{v}}(T,p) + \ub_{I,\vec{v}}(\items\backslash T,\vec{q})\right\}\right].
\end{align}

It will also be useful to introduce notation for the difference in utility from choosing between the principal and the item sellers for the same set of items.
Namely, we define
{
\begin{align} 
\label{eq:diff}
\diff_{\vec{v}}(T, p, \vec{q}) = \ub_{P, \vec{v}}(T, p) - \ub_{I, \vec{v}}(T, \vec{q}).
\end{align}
}

An immediate implication of \eqref{eq:buyer.max} is that given $\vec{v}$, $p$, and $\vec{q}$, the set $T$ of items that the buyer purchases from the principal is a maximizer of $\diff_{\vec{v}}(T,p,\vec{q})$.  The following lemma uses this observation to characterize how a change in the valuation profile $\vec{v}$ can influence the set of items purchased from the principal; this will be useful later when analyzing the principal's optimal revenue.
\begin{restatable}{claim}{claimValueChange}
\label{claim:value-change}
    Fix any $p$ and any $\vec{q}$.  Suppose that $\vec{v}$ and $\vec{v}'$ differ only on a subset $S \subseteq M$ of items, and let $T$ and $T'$ denote the items purchased from the principal under valuations $\vec{v}$ and $\vec{v}'$ respectively.  Then if $T \cap S = T' \cap S$ then $T = T'$. 
\end{restatable}

\subsection{Item Seller Best-Responses} \label{subsec:itemSeller}
{We next turn to the pricing problem faced by an item seller in the game defined by the principal's chosen menu $p$.  We first consider the best response of an item seller $i$ {given a deterministic profile of prices $q_{-i}$}  
selected by the other item sellers.  We will show (in Lemma~\ref{lem:sellerUtility2}) that the best response problem for $i$ can be viewed as a monopolist's pricing problem, but with revenue reduced by a factor that depends on the selected price, with higher prices leading to steeper reductions.

{We begin by characterizing the ways in which a change in one item seller's price can influence the set of items purchased from the principal seller.}
\begin{restatable}{claim}{claimSameMaxMenu}
\label{claim:sameMaxMenu}
Fix some principal menu $p$, item seller $i$, the prices of the other item sellers $q_{-i}$, 
and the realized buyer valuation $\vec v$. Then there exist sets $T_i \ni i$ and $T_{\neg i} \not\ni i$ and a threshold $\theta_i \geq 0$ such that, if $q_i$ is the price chosen by agent $i$, {then the set that the buyer purchases from the principal is $T_i$ if $q_i \leq \theta_i$, and is $T_{\neg i}$ if $q_i > \theta_i$.  } 
\end{restatable}

An immediate corollary of Claim~\ref{claim:sameMaxMenu} is that the probability of sale for an item seller is weakly decreasing in the choice of price $q_i$.  This is true even if we fix the buyer's realized values $\vec{v}$, and in particular even if we condition on the buyer's value for item $i$ being higher than the chosen price $q_i$. This is because a higher price from item seller $i$ will increase the attractiveness of purchasing a set that includes item $i$ from the principal.  To state this formally, 
we will write $\ind_{i, \vec{v}}(p, \vec{q})$ for the indicator variable of sale for item seller $i$, given valuations $\vec{v}$, principal menu $p$, and item seller prices $\vec{q}$.

\begin{restatable}{proposition}{claimProbMonotone} \label{claim:probMonotone}
    Fix any item seller $i$ and prices $q'_i > q_i$ for seller $i$.  Then for any principal menu $p$, other item sellers' prices $q_{-i}$, and buyer values $\vec{v}$, it holds that $\ind_{i, \vec{v}}(p, (q_i, q_{-i})) \geq \ind_{i, \vec{v}}(p,(q'_i, q_{-i}))$.   
\end{restatable}

Our goal now is to analyze the revenue maximization problem for item seller $i$. 
For an item seller $i$, their expected revenue from choosing a price $q_i$ is simply $q_i$ times the probability of sale. 
We will use $\us_{i,\dist}(p, \vec{s})$ to denote the utility of item seller $i$, given the price menu $p$ of the principal, {the value distribution $\dist=\times_i \dist_i$ for the buyer,} and the (possibly mixed) strategies {$\vec{s}$} of all item sellers.
When $\dist$ is clear from the context we sometimes use $\us_{i}(p, \vec{s})$ to denote $\us_{i,\dist}(p, \vec{s})$.
Then, recalling that $\ind_{i, \vec{v}}(p, \vec{q})$ is the indicator variable for sale by item seller $i$, the following lemma relates the item seller's utility to the revenue of a monopolist seller.

\begin{lemma} 
\label{lem:sellerUtility2} For any menu $p$ and any profile of mixed strategies $\vec{s}$ it holds that:
    \begin{align}
    \us_{i} (p, \vec{s}) = 
    {\E_{q_i \sim s_i}\sq*{ \rev_{\dist_i}(q_i)  \cdot \E_{v_{-i} \sim \dist_{-i}, q_{-i} \sim s_{-i}} \sq*{ \ind_{i, (\infty, v_{-i})}(p, \vec{q})}}}
    \label{eq:sellerUtility2}
\end{align}
\end{lemma}

\begin{proof}
From the definition of $\ind_{i, \vec{v}}(p, \vec{q})$, we have
\begin{align}
    \us_{i} (p, \vec{s}) &= 
    \E_{\vec{q} \sim \vec{s}, \vec{v} \sim \dist} \sq*{q_i \cdot \ind_{i, \vec{v}}(p, \vec{q})} 
    &= \E_{\vec{q} \sim \vec{s}, v_{-i} \sim \dist_{-i}} \sq*{q_i \cdot \Prob_{v_i \sim \dist_i} [v_i \geq q_i] \cdot \E_{v_i \sim \dist_i | {v_i}\geq q_i} [\ind_{i, \vec{v}}(p, \vec{q})]} \label{eq:sellerUtility}
\end{align}

We next notice that, taking $T_i$ and $T_{\neg i}$ as in the statement of Claim~\ref{claim:sameMaxMenu}, for any $v_i \geq q_i$ we have that $\ind_{i, \vec{v}}(p, \vec{q}) = 1$ if and only if the buyer does not choose option $T_i$ at price $q_i$, which recall occurs if and only if $\diff_{\vec{v}}(T_i, p, \vec{q}) \geq \diff_{\vec{v}}(T_{\neg i}, p, \vec{q})$ (from equation \eqref{eq:diff} and the discussion immediately following.)
But we notice that, as we change $v_i$, the value of $\diff_{(v_i, v_{-i})}(T_i, p, \vec{q})$ remains the same for any $v_i \geq q_i$.  This is because 
\begin{align*}
    \diff_{(v_i, v_{-i})}(T_i, p, \vec{q}) &= \ub_{P, \vec{v}}(T_i, p) - \ub_{I, \vec{v}}(T_i, \vec{q})\\
    &= \left(\sum_{j \in T_i} v_j - p(T_i)\right) - \left(\sum_{j \in T_i} v_j - \sum_{j \in T_i} \min(v_j,q_j) \right)\\
    &= \sum_{j \in T_i} \min(v_j,q_j) - p(T_i).
\end{align*}
We therefore conclude that $\ind_{i, \vec{v}}(p, \vec{q})$ is a constant for all $v_i \geq q_i$.  In particular, $\E_{v_i \sim \dist_i | \mbe{v_i}\geq q_i} [\ind_{i, \vec{v}}(p, \vec{q})] = \ind_{i, (\infty, v_{-i})}(p, \vec{q})$.  Plugging this into \eqref{eq:sellerUtility}, we conclude that
\begin{align*}
    \us_{i} (p, \vec{s}) &= \E_{\vec{q} \sim \vec{s}, v_{-i} \sim \dist_{-i}} \sq*{q_i \cdot \Prob_{v_i \sim \dist_i} [v_i \geq q_i] \cdot \ind_{i, (\infty, v_{-i})}(p, \vec{q})} \\
    &=\E_{q_i \sim s_i}\sq*{ \rev_{\dist_i}(q_i)  \cdot \E_{v_{-i} \sim \dist_{-i}, q_{-i} \sim s_{-i}} \sq*{ \ind_{i, (\infty, v_{-i})}(p, \vec{q})}}
\end{align*}
   as required.
\end{proof}

{An interpretation of \Cref{lem:sellerUtility2} is that the item seller's revenue when pricing 
{using strategy $s_i$ 
is equal to the expectation (over $q_i\sim s_i$) of the product of the item seller's monopolist Myerson revenue when pricing at $q_i$, reduced by the expectation of $\ind_{i, (\infty, v_{-i})}(p, \vec{q})$ (over $v_{-i}$ and $q_{-i}$) which depends on $q_i$ but not on $v_i$, {and is monotone non-increasing in $q_i$}.} 

\subsection{An Upper Bound on the Principal's Revenue}\label{subsec:upperbound}

In this section we present an upper bound on the principal revenue {for any menu and} in any Nash equilibrium. The upper bound is the expected welfare when each item value distribution is truncated at the {highest} Myerson price for that distribution.

\begin{definition}[Expected Truncated Social Welfare] \label{def:truncated}
    For any $i\in \items$, let $\dist_i$ be a value distribution for item $i$ and let $r_i$ be the maximal Myerson price for $\dist_i$.
    The buyer's \emph{expected truncated social welfare} when her item value vector $(v_1,\ldots, v_n)$ is sampled from $\times_i \dist_i$ 
    is defined to be $\sum_{i} \E_{v_i \sim \dist_i} [\min(r_i, v_i)]$. 
\end{definition}
We prove that the expected truncated  welfare upper bounds the revenue of the principal seller:

\begin{theorem}\label{thm:upper-bound}
Consider a market with value distribution $\dist=\times_i \dist_i$. 
For any menu $p$ of the principal seller, and for any mixed Nash equilibrium $\vec{s}$ in the game $G_{p,\dist}$, the expected revenue of principal seller under $(p,\vec{s})$ is at most the buyer's expected truncated social welfare $\sum_{i} \E_{v_i \sim \dist_i} [\min(r_i, v_i)]$. 
\end{theorem} 

To prove Theorem~\ref{thm:upper-bound}, we first argue that the claimed bound holds under the assumption that each $s_i$ is supported on $[0,r_i]$.  We actually prove something stronger: if each $s_i$ is supported on some arbitrary range $[0,\bar{s}_i]$, then the principal's revenue is at most the welfare when the value of item $i$ is truncated at $\bar{s}_i$.  This is true even if the strategies $s_i$ are not in equilibrium, and even if the principal can choose menu $p$ after the item sellers commit to their strategies.

\begin{proposition}
\label{lem:revenue.supremum.bound}
    Consider a market  with value distribution $\dist=\times_i \dist_i$.
    Fix any menu $p$ of the principal seller and any strategies $\vec{s}$ (which may not be a NE) in the game $G_{p,\dist}$. 
    Denote the supremum of the support of $s_i$ by $\bar{s}_i$. 
    Then the expected revenue of principal seller under $(p,\vec{s})$ is at most  $\sum_{i} \E_{v_i \sim \dist_i} [\min(\bar{s}_i, v_i)]$. 
\end{proposition}
\begin{proof}
    Fix any $\vec{q}$ and $p$.
    Recall that for all $i$ and all $\vec{v}$ in the support of $\dist$, $\ind_{i,\vec{v}}(p,\vec{q})$ is constant for all $v_i > q_i$.  By Claim~\ref{claim:value-change}, this means that the set of items purchased from the principal is constant for all $v_i > q_i$ as well.  Let $\dist'_i$ denote the distribution over $\min\{v_i, q_i\}$ where $v_i \sim F_i$.
    Then we conclude that the distribution over the set of items purchased from the principal is not affected if we replace $\dist_i$ with $\dist'_i$.  Applying this argument to each $i$ in sequence, we conclude that the distribution over the set of items purchased from the principal is identical under $\times_i \dist_i$ and $\times_i \dist'_i$.  This implies that the principal's expected revenue is likewise the same under $\times_i \dist_i$ and $\times_i \dist'_i$.

    Since, for any $p$ and $\vec{q}$, the expected revenue achieved by the principal is at most the expected welfare, we conclude that the principal's revenue is at most the expected welfare under $\times_i \dist'_i$, which is 
    {$\sum_i \E_{v_i \sim \dist'_i} [\min(q_i, v_i)] = \sum_i \E_{v_i \sim \dist_i} [\min(q_i, v_i)]$.} 
    
    Taking an expectation over $\vec{q} \sim \vec{s}$ and noting that $q_i \leq \bar{s}_i$ for all such realizations, the principal's expected revenue is at most $\sum_{i} \E_{v_i \sim \dist_i} [\min(\bar{s}_i, v_i)]$, as claimed.
\end{proof}

Our next step is to consider items $i$ for which the item seller's strategy $s_i$ is not supported on $[0,r_i]$.  We observe that any such seller must be generating revenue $0$ at equilibrium.

\begin{lemma}
\label{lem:revenue.zero} 
    Consider a market with value distribution $\dist=\times_i \dist_i$. 
    Fix any menu $p$ of the principal seller, and any strategy profile $\vec{s}$ in the game $G_{p,\dist}$, such that $s_i$ is a best response of item seller $i$ to $s_{-i}$. If  the probability that $s_i$ assigns to prices strictly greater than $r_i$ is positive,  then the expected revenue of item seller $i$ under $(p,\vec{s})$ in $G_{p,\dist}$ is $0$. 
\end{lemma}
\begin{proof}
For any $z_i > {r_i}$ we have $r_i \times \Prob[v_i \geq r_i] > z_i \times \Prob[v_i \geq z_i]$ (from the definition of $r_i$).  Since we also have $\E_{q_{-i} \sim s_-i, v_{-i} \sim \dist_{-i}} \sq*{\ind_{i, (\infty, v_{-i})}(p,(z_i, q_{-i}))} \leq \E_{q_{-i} \sim s_-i, v_{-i} \sim \dist_{-i}} \sq*{\ind_{i, (\infty, v_{-i})}(p,(r_i, q_{-i}))}$ by Proposition~\ref{claim:probMonotone}, Equation \eqref{eq:sellerUtility2} implies that $\us_{i} (p, (z_i, s_{-i})) < \us_{i} (p, (r_i, s_{-i}))$ unless $\us_{i} (p, (z_i, s_{-i})) = \us_{i} (p, (r_i, s_{-i})) = 0$.  Therefore, if any $z_i > r_i$ is in the support of $s_i$, we must have $\us_{i} (p, (z_i, s_{-i})) = 0$ and hence $\us_{i} (p, (s_i, s_{-i})) = 0$.
\end{proof}

We are now ready to complete the proof of Theorem~\ref{thm:upper-bound}.  Roughly speaking, we will show that the principal cannot extract high revenue from any item $i$ for which the item seller's strategy $s_i$ is not supported on $[0,r_i]$.  This follows because, since any such item seller must be obtaining $0$ revenue at equilibrium, the principal must be selling item $i$ whenever $v_i > 0$.

\begin{proof}[Proof of Theorem~\ref{thm:upper-bound}]
    Observe that, for each $i$, we have 
    ${r}_i>0$,   
    ${r}_i$ is in the support of $\dist_i$, and that  $E_{v_i \sim \dist_i}[\min(v_i, {r}_i)]>0$.
Let $Z$ be the set of indexes of item sellers that price strictly above ${r_i}$ with positive probability according to $s_i$.
By Lemma~\ref{lem:revenue.zero}, each $i \in Z$ has expected revenue $0$ at equilibrium.
    So for any $i \in Z$ and any $q_i > 0$,
    the probability that item seller $i$ sells item $i$ at price $q_i$ must be $0$
    (as otherwise there is a beneficial deviation for item seller $i$). By Claim \ref{claim:sameMaxMenu}, the threshold $\theta_i$ for $i$ must be $0$, and thus there is a set $T_{\neg i}$ that is bought for any $\vec{v}$ (in the support of $\dist$) with $v_i > 0$.
    In particular, for every $\vec{v}$ with $v_i > 0$ the probability (over $\vec{q} \sim \vec{s}$) of the principal selling item $i$ must be $1$ (since otherwise item seller $i$ could obtain positive revenue by pricing at some $\varepsilon_i > 0$). 

    Let us now introduce some notation.  Fixing $p$ and $\vec{q}$, let $T(\vec{v})$ be the set of items purchased from the principal when the valuation profile is $\vec{v}$.  And for any $\vec{v}$, let $\psi(\vec{v})$ be the vector with $i$th coordinate $\min\{v_i, r_i\}$ for all $i \in Z$.  That is, $\psi(\vec{v})$ is the vector $\vec{v}$ with each coordinate $i \in Z$ capped at $r_i$.
    
    We now claim that for any realization of $\vec{q} \sim \vec{s}$ and any $\vec{v}$ in the support of $\dist$, $T(\vec{v}) = T(\psi(\vec{v}))$.  This is because $\vec{v}$ and $\psi(\vec{v})$ differ only on items $i \in Z$ for which $v_i > 0$, and we have argued that $T(\vec{v})$ and $T(\psi(\vec{v}))$ must both contain all such items.  Claim~\ref{claim:value-change} therefore implies that $T(\vec{v}) = T(\psi(\vec{v}))$.

    Write $\dist' = \times_i \dist'_i$ for the distribution over $\psi(\vec{v})$ for $\vec{v} \sim \dist$.
    Since $T(\vec{v}) = T(\psi(\vec{v}))$ for all $\vec{q} \sim \vec{s}$ and $\vec{v} \sim F$, the principal generates the same expected revenue under value distribution $\dist$ and value distribution $\dist'$.  It therefore suffices to show that the principal's expected revenue under $(p,\vec{s})$ with value distribution $\dist'$ is at most the buyer's expected truncated social welfare.

    Define $Y = M \backslash Z$ for notational convenience, so that $Z$ and $Y$ form a partition of $M$.  Write $Rev(Y)$ for the expected revenue of the principal's revenue-optimal menu that sells only items in $Y$, over $\vec{v} \sim \dist'$ and $\vec{q} \sim \vec{s}$.  Write $Wel(Z)$ for the expected welfare of items in $Z$, again over $\vec{v} \sim \dist'$.  We then claim that the principal's expected revenue under $(p,\vec{s})$ is at most $Rev(Y) + Wel(Z)$. 
    {
        \begin{restatable}{claim}{lemSubsetBound}\label{lem:subset-bound-rev-plus-welfare}
        Consider a market with value distribution $\dist' = \times_i \dist'_i$ and {a partition of the items to sets $Y, Z \subset \items$ ($Y \mathbin{\dot{\cup}} Z=\items$).}
        For any menu $p$ of the principal seller, and for any mixed Nash equilibrium $\vec{s}$ in the game $G_{p,\dist}$, the expected revenue of principal seller under $(p,\vec{s})$ is at most $Rev(Y) + Wel(Z)$.  
        \end{restatable}  
    }  
    We note that such a bound on the seller's revenue is already known for the case of a monopolist (see~\cite{HartN17}), and the proof for the principal seller is very similar; {we include the proof of \Cref{lem:subset-bound-rev-plus-welfare} in \Cref{app:sec:upperbound} for completeness. }
Now, by \Cref{lem:revenue.supremum.bound}, the definition of $Y$, and the fact that $\dist'_i = \dist_i$ for all $i \in Y$, we have $Rev(Y) \leq \sum_{i \in Y} \E_{v_i \sim \dist_i} [\min(r_i, v_i)]$.  From the definition of $\dist'_i$ for $i \in Z$, $Wel(Z) = \sum_{i \in Z} \E_{v_i \sim \dist_i} [\min(r_i, v_i)]$.  Our claim then implies that the principal's expected revenue under $(p,\vec{s})$ is at most $\sum_i \E_{v_i \sim \dist_i} [\min(r_i, v_i)]$, as required.
\end{proof}

We note that there are markets in which the principal cannot obtain revenue that is any constant approximation to the upper bound: 
a trivial example illustrating this is a market with only one item, in which the principal cannot obtain positive revenue, yet the expected truncated social welfare is positive.  Yet, when the principal is able to find a menu with revenue that is constant approximation to the upper bound, it is clearly a constant approximation to the optimal revenue of the principal. In the next section we move to present sufficient conditions for such revenue extraction.  

\section{Approximating Optimal Revenue via Bundle Pricing} \label{sec:approxopt}

In this section we show that under certain conditions on the value distributions, the principal seller can guarantee expected revenue that is a constant approximation to the expected truncated social welfare.  Moreover, this can be achieved by a menu that prices only the grand bundle of all items.  As the expected truncated social welfare is an upper bound on the principal's revenue (by \Cref{thm:upper-bound}), that simple menu is a constant approximation to the revenue-maximizing menu for the principal. 

{
We begin by building up some intuition for the conditions we will impose on the value distributions.
Intuitively, selling the items as a single grand bundle can potentially 
be a good strategy for the principal when the buyer's value for the grand bundle is concentrated (e.g., close to a Gaussian distribution), and no single item affects the buyer's value for the bundle to a significant degree. 
Such a concentration plays an important role 
even in the simpler setting of Monopolist pricing \cite{BabaioffILW14}, and is guaranteed to occur when the variance of the buyer's value for the grand bundle is much larger than the buyer's value for any individual item.
}

{
However, in our setting, the principal seller faces extra complexity due to the competing item sellers. 
Our goal is a revenue guarantee that holds over all equilibria (similar to the price of anarchy),
so we must bound the principal's revenue over all outcomes that can occur in item seller equilibria.  Ideally (from the principal's perspective), each item seller would set a price that is not too much lower than their highest Myerson price.
To this end, we further analyze item seller best-response strategies by examining the revenue curve faced by an item seller. Recall from \Cref{lem:sellerUtility2} 
that, given the principal's pricing menu $p$, item seller $i$'s revenue {under strategy profile}
$(s_1, \cdots, s_n)$ 
is equal to 
\begin{align*}
\us_{i} (p, \vec{s}) &=
{\E_{q_i \sim s_i}\sq*{ \rev_{\dist_i}(q_i)  \cdot \E_{v_{-i} \sim \dist_{-i}, q_{-i} \sim s_{-i}} \sq*{ \ind_{i, (\infty, v_{-i})}(p, \vec{q})}}} \\
      & = \E_{q_i \sim s_i}\sq*{ \rev_{\dist_i}(q_i)  \cdot \Prob_{v_{-i} \sim \dist_{-i}, q_{-i} \sim s_{-i}}\sq*{\max_{T \ni i} \set*{ p(T) - \sum_{j \in T, j \neq i} \min(v_j, q_j)} \geq q_i}},
\end{align*}
where we recall that  $\rev_{\dist_i}(\cdot)$ 
is the monopolist revenue function for item $i$ with value distribution $\dist_i$. If we restrict our attention to a principal 
who is using a grand bundling strategy, where $p(T) = p(\items)$ for all non-empty sets $T$, 
the above formula can be simplified to 

\begin{align}
    \us_{i} (p, \vec{s}) = \E_{q_i \sim s_i}\sq*{ \rev_{\dist_i}(q_i)  \cdot \Prob_{v_{-i} \sim \dist_{-i}, q_{-i} \sim s_{-i}}\sq*{p(\items) - q_i \geq \sum_{j \in \items, j \neq i} \min(v_j, q_j)}}. \label{eq:item-seller-util-bundle}
\end{align}

Observe that the first term inside the expectation is just the monopolist revenue of item seller $i$ when pricing at $q_i$, and the second term can be interpreted as the probability that the buyer does not purchase the grand bundle from the principal, 
conditioned on the buyer having value for item $i$ that is at least $q_i$. 
{
As long as the aggregate variance in the truncated social welfare is high, the second term will be relatively insensitive to changes in $q_i$.
If we additionally assume that item seller $i$'s monopolist revenue is non-trivially sensitive to large price reductions (i.e., well below the highest Myerson price), item seller $i$'s best response price must be close to their highest Myerson price, which is what we want.  
This motivates the following property of value distributions, which we call ``price sensitivity." 
}

\begin{definition}[(Monopolist) Price Sensitivity] \label{def:priceSensitivity}
    Let $\dist_i$ be the value distribution of item $i$ and $r_i$ be the highest Myerson price for item $i$. $\dist_i$ is \emph{$(\lambda, C)$-price sensitive} for parameters $C\in (0,1)$ and $\lambda > 0$ 
    if for all $\alpha \in [0, C]$ it holds that $\frac{\rev_{\dist_i}(r_i) - \rev_{\dist_i}(\alpha \cdot r_i)}{(1-\alpha) \cdot r_i} \geq \lambda$. 
\end{definition}

{Intuitively, the price sensitivity property requires that the revenue curve $\rev_{\dist_i}(\cdot)$ decays steeply enough for prices substantially lower than the highest Myerson price $r_i$.  A geometric interpretation of $(\lambda, C)$-price sensitivity is that, over the range $q \in [0, C \cdot r_i]$, the revenue curve $\rev_{\dist_i}(\cdot)$ lies below a line of slope $\lambda$ passing through the point $(r_i, \rev_{\dist_i}(r_i))$.  See Figure~\ref{fig:price-sensitivity}.}

In our main theorem 
(\Cref{thm:generalPrincipalBundleRev})
we formally show that large variance and price sensitivity are sufficient conditions for selling the grand bundle to be constant revenue-approximating for the principal seller. In this section we will provide intuitions and discuss the proof outline for 
\Cref{thm:generalPrincipalBundleRev}.
All missing proofs in this section can be found in \Cref{app:sec:approxopt}.

{We define some notations. For $v_i\sim \dist_i$ we use $V_i(r_i)$ to denote the random variable $\min(v_i, r_i)$: the value of item $i$ truncated at $r_i$. Note that the truncated social welfare is the random variable $V(\vec{r}) = \sum_i V_i(r_i)= \sum_i \min(v_i, r_i)$ where $v_i\sim \dist_i$ for every $i$ and $\vec{r}=(r_1,...,r_m)$. 
The expected truncated social welfare is simply the mean of $V(\vec{r})$, and we write $\sigma(V(\vec{r}))$ to denote the standard deviation of 
$V(\vec{r})$. 
We use $\mu_i(r_i) = \mu(V_i(r_i))$ to denote the mean of $V_i(r_i)$, and $\sigma_i(r_i) = \sigma(V_i(r_i))$ to denote the standard deviation of  $V_i(r_i)$.
For distributions $\dist_1,\ldots, \dist_m$ we define $K(\dist_1,\ldots, \dist_m) = \frac{\max_{i \in \items} \set*{r_i - \mu_i(r_i)}}{\min_{i \in \items} \set*{r_i - \mu_i(r_i)}}$ to be the ratio between the largest and smallest value of $\set*{r_i - \mu_i(r_i)}$. Also, for notational convenience, we will sometimes use $p$ to denote the principal's price of the grand bundle, instead of $p(\items)$.  
}
We are now ready to state our main result.

\begin{restatable}{theorem}{thmGeneralPrincipalBundleRev}\label{thm:generalPrincipalBundleRev}
Consider a market with value distribution $\dist=\times_i \dist_i$. 
Assume that there exists constant $\lambda>0$ 
and $C = 1 - \frac{\min_{i \in \items} F_i(r_i)^4}{8K+1}$
where 
$K= K(\dist_1,\ldots, \dist_m)$
such that the value distribution of every item $i \in \items$ is $(\lambda, C)$-price sensitive. In addition, assume 
$\sigma(V(\vec{r})) \geq \frac{12}{(\lambda(1-C))^{3/2}} \cdot \max_{j \in \items} \{r_j\}$. 
When the principal offers the grand bundle at price $p = \sum_{i}\E[\min(v_i, C \cdot r_i)] + \frac{\sigma(V(\vec{r}))}{4}$, the revenue of the  principal  in every Nash equilibrium of the game $G_{p,\dist}$ 
is  at least $\frac{1}{3}$ of the expected truncated social welfare $\sum_{i}\E[\min(v_i, r_i)]$.
\end{restatable}

We note that as $K= K(\dist_1,\ldots, \dist_m)$ grows large, $C$ approaches $1$ and the requirement on the variance $\sigma(V(\vec{r}))$ becomes more demanding. Thus, conditions that imply that $K$ is not too large are desirable. 
{For example, when items are i.i.d. we get $K=1$, the best we can hope for. For non-i.i.d., it will sometimes be beneficial for the principal not to bundle items that significantly increase $K$, when their removal does not decrease the expected truncated welfare too much (such pre-processing is simple, and we do not discuss it further).}

\subsection{Implications} \label{subsec:approxopt:implications}
Since the statement of \Cref{thm:generalPrincipalBundleRev} is a bit technical, we here discuss two 
corollaries of
\Cref{thm:generalPrincipalBundleRev} that provide simpler sufficient conditions for our approximation result.
In both cases we show that 
selling the grand bundle gives a constant approximation to the principal's optimal revenue when the number of items is sufficiently large.  In the first case, 
items values are sampled i.i.d.\ from a distribution for which the Myerson price is not the minimal value in the support.  In the second case, 
each (not necessarily identical) distribution is  $\delta$-smooth and $\delta$-revenue-concave (formally defined in \Cref{def:rev-concave} and \Cref{def:dist-smooth}).

\subsubsection{I.I.D. Valuations}
We start with the i.i.d. case, first looking at a single distribution.
\begin{restatable}[One Item Distribution, price sensitivity]{lemma}{lemPriceSense}\label{lem:price-sense}
   Consider some value distribution $G$, and assume it has a highest 
   Myerson price $r_G$. Then $G$ is $(\lambda, C)$-price sensitive for $C = 1 - \frac{G(r_G)^4}{9}$ and $\displaystyle \lambda = \lambda(G) = \min_{C' \leq C} \frac{\rev_{G}(r_G) - \rev_{G}(C' \cdot r_G)}{(1-C') \cdot r_G}$. 
\end{restatable}

Let $\sigma_G(r_G)$ 
denote the standard deviation of $\min(v, r_G)$, where $v \sim G$.  
We next show that for a distribution $G$ that is $(\lambda, C)$-price sensitive and has a highest 
Myerson price $r_G$, 
the sum of $m$ i.i.d. samples from distribution $G$ will reach a target minimal variance threshold $\frac{12}{(\lambda(1-C))^{3/2}} \cdot r_G$ if we take $m$ that is large enough. 
\begin{restatable}[Independent and Identical Items, variance]{lemma}{lemIIDVar}\label{lem:iid-var}
    Consider markets with $m$ i.i.d items, each with value distribution $\dist_i = G$ that is $(\lambda, C)$-price sensitive and has a highest 
    Myerson price $r_G$. 
    If $m \geq \frac{144}{\lambda^3(1-C)^3} \cdot \paren*{\frac{r_G}{\sigma_G(r_G)}}^2$ then $\sigma(V(\vec{r})) \geq \frac{12}{(\lambda(1-C))^{3/2}} \cdot r_G$. 
\end{restatable}

Combining \Cref{lem:price-sense}, \Cref{lem:iid-var} and \Cref{thm:generalPrincipalBundleRev} we derive the following corollary.

\begin{restatable}[Independent and Identical Items, revenue]{corollary}{corIIDRev}\label{cor:iidRev}
Let $G$ be a distribution that has a unique Myerson price $r_G>inf (G)$.
Then for 
all sufficiently large $m$ {(polynomial in parameters of $G$)} 
and valuation distribution $\dist=\times_{i=1}^m G= G^m$, 
the principal can find a price for the grand bundle such that she obtains revenue at least a $\frac{1}{3}$ fraction of the expected truncated social welfare $\sum_{i}\E[\min(v_i, r_i)]$ in every equilibrium of $G_{p,\dist}$. 
\end{restatable}

\subsubsection{Non-I.I.D. Smooth and Revenue-Concave Valuations}

For non-i.i.d.\ distributions, a similar result can be obtained when all item distributions satisfy the following extra conditions. 
The following notion captures a function with negative second derivative bounded away from $0$. 

\begin{definition}[$\delta$-revenue-concave] \label{def:rev-concave}
    For $\delta>0$, we say that a 
    distribution $G$ is \emph{$\delta$-revenue-concave} if $G$ is twice differentiable, and for revenue function $h(x) := \rev_G(x) = x \cdot (1 - G(x))$, it holds that $h''(x) \leq -\delta/r_G$ for all $x \in (0, r_{G})$, where $r_G$ is the highest Myerson price of $G$.\footnote{Note that $\delta$-revenue-concavity implies that the Myerson price is unique.} 
\end{definition}

The following notion captures a smooth distribution, where the probability density function is bounded away from $0$ from below.

\begin{definition}[$\delta$-smooth] \label{def:dist-smooth}
    For $\delta>0$, we say that a distribution $G$ is \emph{$\delta$-smooth} if {$G$ is differentiable and $G'(x) \geq \delta/r_G$ for every $x \in (0, r_G)$, where $r_G$ is the highest Myerson price of $G$. }
\end{definition}

We now show that when a distribution is $\delta$-revenue-concave and $\delta$-smooth, we can lower bound the price sensitivity of the distribution and the variance of truncated welfare, using $\delta$ as a parameter. 

\begin{restatable}[Revenue Concave Distributions, price sensitivity]{lemma}{lemSmoothPriceSense}\label{lem:smooth-price-sense}
   Consider some value distribution $G$ that for some $\delta>0$ is
    $\delta$-revenue-concave. Let $r_G$ be the unique Myerson price of $G$.  
   Then $G$ is $(\lambda, C)$-price sensitive for {$\lambda \leq \frac{\delta}{2} \cdot (1 - C)$.}
\end{restatable} 

\begin{restatable}[Smooth Distributions, variance]{lemma}{lemSmoothVariance}\label{lem:smooth-variance}
    Consider some value distribution $G$ supported on $[0,1]$ 
    that is $\delta$-smooth with a highest Myerson price $r_G$.
    Then the variance of $\min(v, r_G)$, where $v \sim G$, is   
    $\Omega(\delta) \cdot (r_G)^2$.     
    Moreover, $r_G - \E_{v \sim G}[\min(v, r_G)] = \Omega(\delta) \cdot r_G$. 
\end{restatable}
\Cref{lem:smooth-variance} shows that the standard deviation of individual items whose distribution is $\delta$-smooth 
is bounded by $\Omega(\delta)$. Meanwhile, a more detailed reasoning about \Cref{lem:smooth-price-sense} (see a full discussion in \Cref{app:subsec:approx-implication}) shows that in order to apply \Cref{thm:generalPrincipalBundleRev}, it is sufficient for the standard deviation for the buyer's value of all items to be $\Theta\left(\mathsf{poly}\left(\frac{1}{\delta}, \frac{r_{max}}{r_{min}}\right)\right)$, 
where $0 < r_{min} \leq r_{max}$ are values such that the interval $[r_{min}, r_{max}]$ contains each item's highest Myerson price. 
Taking the number of items $m$ to be sufficiently large ($m=\mathsf{poly}\left(\frac{1}{\delta}, \frac{r_{max}}{r_{min}}\right)$) ensures the standard deviation constraint in \Cref{thm:generalPrincipalBundleRev} is satisfied, and yields the following corollary. 
\begin{restatable}[Revenue Concave Smooth Distributions, revenue]{corollary}{corSmoothRev} \label{cor:smooth-rev}
Fix any $r_{max}$, $r_{min}> {0}$. 
Consider a market with $m$ items, where each item $i$ has a value distribution $\dist_i$ that is $\delta$-smooth and $\delta$-revenue-concave, with highest Myerson price $r_i \in [r_{min}, r_{max}]$. 
When {$m = \Omega(\mathsf{poly}(\frac{1}{\delta}, \frac{r_{max}}{r_{min}}))$}, 
the principal can sell the grand bundle and obtain revenue at least a $\frac{1}{3}$ fraction of the expected truncated social welfare $\sum_{i}\E[\min(v_i, r_i)]$ in every equilibrium of $G_{p,\dist}$.
\end{restatable}
\subsection{{Outline of the Proof of Theorem \ref{thm:generalPrincipalBundleRev}}} \label{subsec:approxopt:proof} 
\paragraph{Additional Notations. } {Expanding on the notation of $V_i(r_i)$ that we defined earlier in the section, we will use $V_i(s_i)$ to denote the random variable $\min(v_i, q_i)$, where {$v_i\sim \dist_i $ and} $q_i$ is a random pricing strategy that is drawn from the item seller $i$'s mixed strategy $s_i$. We will denote $\mu_i(s_i) = \mu(V_i(s_i))$ and $\sigma_i(s_i) = \sigma(V_i(s_i))$ as the mean and standard deviation of $V_i(s_i)$, respectively. Similarly, we will use $\mu(\vec{s})$ and $\sigma(\vec{s})$. to represent the mean and standard deviation of $\sum_{j \in \items} V_j(s_j)$, respectively. We also use $\mu_{-i}(s_{-i})$ and $\sigma_{-i}(s_{-i})$ to represent the mean and standard deviation of $\sum_{j \neq i} V_j(s_j)$, respectively. For the sake of simplicity in notation, when item seller $i$'s strategy $s_i$ is pure and pricing at $q_i$, we would just use $q_i$ to represent $s_i$. For example, $V_i(C \cdot r_i)$ denote the random variable of $\min(v_i, C \cdot r_i)$ when $v_i\sim F_i$. }

\paragraph{{Preliminary: The Berry-Esseen Theorem}}In our proof of the main theorem, we will need to obtain a tight bound on the following probability term in \Cref{eq:item-seller-util-bundle}, which can be viewed as the c.d.f of the sum of $m-1$ independent random variables $V_j(s_j)$. 
\begin{align}
    \Prob_{v_{-i} \sim \dist_{-i}, q_{-i} \sim s_{-i}}\sq*{p(\items) - q_i \geq \sum_{j \in \items, j \neq i} \min(v_j, q_j)} = \Prob_{v_{-i} \sim \dist_{-i}}\sq*{p(\items) - q_i \geq \sum_{j \in \items, j \neq i} V_j(s_j)}. \label{eq:item-seller-prob-term2}
\end{align}
To bound the difference between the distribution of $\sum_{j \in \items, j \neq i} V_j(s_j)$ and that of a Gaussian, we will make extensive use of the Berry-Esseen Theorem, which is stated below. 
\begin{theorem}[Berry-Esseen Theorem~\cite{berry1941accuracy,esseen1942liapunov}]\label{claim:BerryEsseenGeneral}
    Given independent random variables $Y_1, \cdots, Y_m$ with mean $0$, for any index $j \in [m]$, let $\sigma_i$ be the standard deviation of $Y_i$, and let $\rho_{i} = \E[|Y_i|^3]$. Let $\sigma$ be the standard deviation of $\sum_{i=1}^m Y_i$ ($\sigma^2 = \sum_{i=1}^m \sigma_i^2$).  Then the Kolmogorov–Smirnov distance between $\frac{\sum_{i} Y_i}{\sigma}$ and the standard normal distribution $\N(0, 1)$ is at most $0.5606 \cdot \frac{1}{\sigma} \cdot \max_{i=1}^m \frac{\rho_i}{\sigma^2_i}$. 
\end{theorem}
Note that since the Berry-Esseen Theorem requires the random variables of mean $0$, the theorem will be applied on random variables $Y_j = V_j(s_j) - \mu_{j}(s_{j})$. The probability term in \Cref{eq:item-seller-prob-term2} remain unchanged when $\sum_{j \neq i} \mu_{j}(s_{j})$ is subtracted from both sides of the inequality ``$p(\items) - q_i \geq \sum_{j \in \items, j \neq i} V_j(s_j)$''. 

\paragraph{Preserving Variance Between $\min(v_i, r_i)$ and $V_i(s_i)$. }
When proving our main theorem, we will assume that the buyer's value for the bundle has a much higher variance than the highest Myerson price for any item seller $i$. 
However, this does not immediately translate into an understanding of the variance of $\sum_{j \in \items, j \neq i} V_j(s_j)$, which is the quantity that we actually need to study. 
{Here} we present sufficient conditions for the mean and variance of  $\sum_{j \in \items} V_j(s_j)$ to be within a constant factor of the mean and variance of the random variable $\sum_{j \in \items} \min(v_j, r_j)$. 
{This will enable us to lower bound the mean and variance of $\sum_{j \in \items, j \neq i} V_j(s_j)$ (which is sufficiently close to that of $\sum_{j \in \items} V_j(s_j)$). 
As a result we obtain tight bounds on probability distribution of  the random variable} 
$\sum_{j \in \items, j \neq i} V_j(s_j)$ using the Berry-Esseen Theorem. 

Firstly, when $s_j$ is lower bounded by $C \cdot r_i$ for sufficiently large $C$ (as specified below), Lemma~\ref{coro:constVarianceBound} implies that the variance of $V_j(s_j)$ will be a constant fraction of the variance of $\min(v_j, r_j)$.  

\begin{restatable}{lemma}{coroConstVarianceBound}  \label{coro:constVarianceBound}
     For any $K \geq 1$ and $C$ such that $C \geq 1 - \frac{\dist_i(r_i)^4}{2K+1}$, 
     it holds that $\sigma_i^2(C \cdot r_i) \geq \paren*{1 - \frac{1}{K}} \cdot\sigma_i^2(r_i)$. 
\end{restatable}
Even if \Cref{coro:constVarianceBound} is not satisfied for some items, as long as the mean of $\sum_{j \in \items } V_j(s_j)$ is high enough, the variance of $\sum_{j \in \items } V_j(s_j)$ will also constant approximate the variance of the buyer's truncated welfare.

\begin{restatable}{lemma}{lemHighMeanVar} \label{lem:highMeanVar}
    For any $C \geq 1 - \frac{\min_{i \in \items} F_i(r_i)^4}{8K+1}$, where $K = K(\dist_1,\ldots, \dist_m)$ and any item seller strategy profile such that $\mu(\vec{s}) = \sum_{i} \mu_i(s_i) \geq \sum_{i} \mu_i(C \cdot r_i)$, it holds that  $\sigma^2(\vec{s}) = \sum_i \sigma_i^2(s_i) \geq \frac{1}{2} \sum_{i} \sigma^2_{i}(r_i) = \sigma^2(\vec{r})$.
\end{restatable}

Now, we are ready to discuss the proof outline of our main theorem. 
\paragraph{Proof Outline of \Cref{thm:generalPrincipalBundleRev}. }
When the principal prices the grand bundle at price $p= \sum_{i} \E[\min(v_i, C \cdot r_i)] + \frac{\sigma(V(\vec{r}))}{4}$,  consider any mixed Nash equilibrium $(s_1, \cdots, s_n)$ between the item sellers in $G_{p,\dist}$. Note that the principal's revenue is just equal to $p$ times the probability that the buyer purchases the grand bundle from the principal. Fix any vector of values $\vec{v}$. The buyer purchases the grand bundle from the principal if and only if 1) their value for the bundle $\sum_{i} v_i$ is at least $p$ and 2) the buyer gains more utility from purchasing all items from the principal seller rather than buying any subset of items from the item sellers. Letting $q_i \sim s_i$ be the realized price set by each item seller $i$, {we show that} 2) is satisfied when $p < \sum_{i} \min(v_i, q_i)$. {(Essentially, this inequality implies that the buyer gains higher utility from purchasing the principal's grand bundle than purchasing any subset of items from the item sellers.)} 
This automatically implies that condition 1) is satisfied as well.  
Hence the principal's revenue is just equal to $p \cdot \Prob_{\vec{v}\sim \dist, \vec{q}\sim \vec{s}}\left[p < \sum_{i} \min(v_i, q_i)\right]$.  

To show that this revenue is high, note that since
$K\geq 1$ (by definition), it holds that $C = 1 - \frac{\min_{i \in \items} F_i(r_i)^4}{8K+1}\geq 1-\frac{1}{8+1} = \frac{8}{9}$. Hence $p \geq C \cdot \sum_{i} \E[\min(v_i, r_i)] \geq \frac{8}{9} \cdot \sum_{i} \E[\min(v_i, r_i)]$.
So if it were the case that $\Prob_{\vec{v}\sim \dist, \vec{q}\sim \vec{s}}[p < \sum_{i} \min(v_i, q_i)] \geq 1/2$, 
then the principal seller achieves a revenue of $\frac{1}{2} \cdot \frac{8}{9} \cdot \sum_{i} \E[\min(v_i, r_i)]$, which is at least $\frac{4}{9}$ of  
our revenue benchmark, as required.
It remains to consider the case $\Prob[p < \sum_{i} \min(v_i, q_i)] < 1/2$.  Our goal is to show we can still achieve a high approximation ratio to our revenue benchmark in this case. We show this by proving the following lemma, which shows that each item seller $i$ would never want to price below $C \cdot r_i$. 
\begin{restatable}{lemma}{lemContainedResponse} \label{lem:containedResponse}
Consider a market with value distribution $\dist=\times_i \dist_i$. 
Assume that there exists constant $\lambda > 0$ 
such that for $C = 1 - \frac{\min_{i \in \items} F_i(r_i)^4}{8K+1}$ where $K =  K(\dist_1,\ldots, \dist_m)$, all items $i \in \items$ have value distributions that are $(\lambda, C)$-price sensitive. In addition, assume the variance of the buyer's value for the bundle $\sigma(V(\vec{r})) \geq \frac{12}{(\lambda(1-C))^{3/2}} \cdot \max_{j \in \items} \{r_j\}$. Given principal grand bundle price $p = \sum_{i}\E[\min(v_i, C \cdot r_i)] + \frac{\sigma(V(\vec{r}))}{4}$ and item seller equilibrium $(s_1, \cdots, s_n)$ in $G_{p,\dist}$, if the buyer not purchases the grand bundle with probability $\geq 1/2$, then the support of each $s_i$ is contained in $[C \cdot r_i, r_i]$. 
\end{restatable}

\Cref{lem:containedResponse} is quite technical, so here we summarize the flow of the argument. First, we show that when the probability of the buyer not purchasing the grand bundle from the principal is at least $1/2$, 
then this probability would still remain relatively large  even if one of the sellers (say seller $i$) changes their pricing strategy arbitrarily. Moreover, we show that unless the probability of the buyer not purchasing the grand bundle is close to $1$, the mean of $\sum_{j \neq i} V_j(s_j)$ is high enough to apply \Cref{lem:highMeanVar}, and hence the variance of $\sum_{j \neq i} V_j(s_j)$ is large enough {for the error term to be sufficiently small when applying the Berry-Esseen Theorem.}
Given this limited competitive pressure from the principal and sufficiently large variance, we apply Berry-Esseen to show that it is strictly suboptimal for an item seller $i$ to price below $C \cdot r_i$. Hence, in the equilibrium $(s_1, \cdots, s_n)$, the support of each item seller $i$'s strategy $s_i$ is contained in $[C \cdot r_i, r_i]$. 

Finally, we apply
\Cref{coro:constVarianceBound} to argue that $\sum_{j \neq i} V_j(s_j)$ has high variance, then use the Berry-Esseen Theorem again (plus the fact that item sellers' equilibrium prices lie in $[C\cdot r_i, r_i]$ from Lemma~\ref{lem:containedResponse}) to prove that the value of $\Prob[p < \sum_{i} \min(v_i, q_i)]$ 
must be at least 0.38.  Since $C\geq 8/9$, we conclude that the principal's revenue is at least $(0.38)C \geq 1/3$ of the expected truncated welfare. 

\section{Bundling is Not Always Revenue Constant Competitive.} \label{sec:brev-prev-lb}

{
Theorem~\ref{thm:generalPrincipalBundleRev} provides conditions under which selling the grand bundle approximates not only the optimal principal's revenue, but our relaxed benchmark of the expected truncated social welfare.  Are such these conditions necessary?  Of course, we cannot hope to approximate the expected truncated social welfare in all settings: for example, if there is only a single good, the logic of Bertrand competition shows that the principal cannot generate positive revenue with any menu.  But even in this case, selling only the grand bundle trivially achieves the principal's optimal revenue (which is $0$).  This leads us to ask: does selling the grand bundle always yield a constant approximation to the principal's optimal revenue?
}
We show in this section that the answer is \emph{no}: there are markets in which the principal seller cannot extract high revenue by pricing the grand bundle, 
yet some other deterministic menu is able to extract much higher revenue. Thus, it is not always the case that by simply pricing the grand bundle the principal can approximate the optimal revenue. 
Specifically, the market constructed does not have high enough variance to match the condition on the variance of Theorem \ref{thm:generalPrincipalBundleRev}, and indeed we show that pricing the grand bundle does no provide high revenue in any NE. On the other hand, the fact that in that market there is no single item that dominates all other items is shown to be enough for some other menu to obtain much higher revenue in every equilibrium (by bundling  similar items and pricing those bundles).

\begin{theorem}\label{thm:prev}
    There exist $m$-item markets with the following properties: 
    For value distribution $\dist=\times_{j=1}^m \dist_j$ of that market, there is a principal menu $p$ such that in any equilibrium strategies $\vec{s}$ for the item sellers in $G_{p,\dist}$, the principal revenue is $\Theta(m)$, while for any grand pricing menu $p'$ and any equilibrium strategies for the item sellers in $G_{p', \dist}$, the revenue of the principal seller is $O(1)$.
\end{theorem} 

{
To provide some intuition behind Theorem~\ref{thm:prev}, recall that in the monopolist setting the revenue-optimal menu can always be constant-approximated by either selling the grand bundle only or selling all items separately~\cite{BabaioffILW14}.  
In that result, the ability to sell all items separately is important to capture revenue in cases where a single item has unexpectedly high realized value.
In contrast, the principal in our setting cannot capture this revenue in the same way: selling items separately generates no revenue, since this results in Bertrand competition for each item.
This doesn't immediately imply Theorem~\ref{thm:prev}, however, because our upper bound on revenue truncates very high values, so the scenario with just a single high value is not problematic for approximating the optimal revenue (i.e., because it is anyway hopeless for the principal to extract substantial revenue from any one item).
However, as it turns out, there can still be situations where the optimal revenue is driven by a few high-valued items (more than one, but at most a constant) that dominate the total welfare.  These cases are necessarily rare, and would be handled by selling items separately in the monopolist setting.  But in our setting, these rare events might be the most important to revenue.  By designing the high-item-value events to occur at different orders of magnitude for different items, one can construct scenarios where the principal cannot obtain high revenue by bundling all items together, but could generate high revenue by bundling together items whose values are at similar scales, essentially simulating the sell-items-separately solution that would be employed by a monopolist.
} 

{We now describe
our construction for Theorem~\ref{thm:prev} that implements the intuition above.}
Fix any constant $K>2$; we will prove
the claim for $m\gg K^2$ that is even. 
Let $n=m/2$.
In the market we construct there are $n$ pairs of items: $T_1, \cdots, T_{n}$. For each $i\in [n]$,  set $T_i$ contains two items, sampled i.i.d from the following binary distribution $\dist_i$.  
The value is $H_i = K^{2 \cdot 3^{(i-1)}}$ with probability $x_i = K^{-3^{(i-1)}}$, and zero otherwise. There are $m=2n$ items: for $i\in [n]$, both item ${2i-1}$ and item ${2i}$ are distributed according to $\dist_i$ (so $\dist= \times_{i=1}^m (F_i\times F_i)$).  
{Observe that for clarity of exposition the construction uses integer values instead of values in $[0,1]$ (the construction can of course be rescaled to $[0,1]$).}

In Section~\ref{sec:prev} we will construct a menu for the principal with revenue $\Omega(m)$; this will be a partition menu that sets a price on each pair of items $T_i$.  Then in Section~\ref{sec:brev-lb} we will show that every grand bundle menu for the principal generates revenue $O(1)$, competing the proof of Theorem~\ref{thm:prev}.

\subsection{Principal Revenue from Partition Bundles.} 
\label{sec:prev}

Consider a principal pricing strategy that bundle items in each set $T_i$ together at price $H_i + \varepsilon = K^{2 \cdot 3^{(i-1)}} + \varepsilon$. We will prove that in any item seller equilibrium, the revenue of the principal is $\Theta(m)$. Intuitively, for each set $T_i$, the principal seller is selling the bundle with probability at most $x_i^2$. Hence the item sellers cannot improve their probability of sale by too much by reducing their price from $H_i$ (which is the unique Myerson price for item $i$). Consequently, each item seller's price in equilibrium will {price at exactly} 
$H_i$. This ensures that when the principal is pricing at $H_i + \varepsilon$, the buyer would purchase the bundle $T_i$ from the principal if 
they have a high value for both items in $T_i$. By our construction of the example, the principal gets revenue at least $H_i \cdot x_i^2 = 1$ for each set $T_i$. Adding the revenue over all  $n$  sets, the principal gets revenue at least $n = m/2$. We formalize this in the next claims.  
\begin{restatable}{claim}{propPerPartitionRev} \label{prop:perPartitionRev}
    Fix $x < \frac{1}{2}$. Let $D$ be the following binary distribution   supported on $\{0, H = \frac{1}{x^2}\}$: the value is $H$ with probability $x$, and $0$ otherwise.
    Consider a market with buyer value sampled from $\dist= D\times D$. 
    Then when the principal sells the bundle of the two items at $p = H + \varepsilon$ for some sufficiently small $\varepsilon$, 
    there is a unique NE (in which each item seller prices at $H$), and the revenue of the principal in this NE is at least $1$.
\end{restatable}

\begin{restatable}{claim}{propPartitionEquilibrium} \label{prop:PartitionEquilibrium}
    Consider a market with buyer value distribution $\dist= \times_{i=1}^m (F_i\times F_i)$. 
    Assume the principal seller is using a partition bundle pricing $p$ over item set $\items$, where the partitioned disjoint sets are $T_1$, $\cdots$, $T_n$. Then $(s_1, \cdots, s_{m})$ is an equilibrium in $G_{p,\dist}$ if and only if for each $i \in [n]$, $\{s_{j}\}_{j \in T_i}$ is an equilibrium for the sub-game $G_{p(T_i), F_i\times F_i}$ over item set $T_i$, where the principal is posting a bundle price $p(T_i)$. 
\end{restatable}

We note that the claim above holds more generally: every partition menu essentially breaks the game between the items sellers to sub-games, one for each part of the partition, with correspondence between NE in the large game and in the sub-games. Combining the two claims we conclude:

\begin{restatable}{lemma}{lemPrevBound} \label{lem:prevBound}
There exists a principal pricing strategy $p$ that sells partitioned bundles 
$T_1$, $\cdots$, $T_n$, such that when each pair $T_j$ is priced at $H_j+\varepsilon$ there is a unique NE in the game  $G_{p,\dist}$ and the principal's revenue of this NE is $\Theta(m)$. 
\end{restatable}

\subsection{Principal Revenue from Selling the Grand Bundle.} 
\label{sec:brev-lb}

We saw that the principal can obtain revenue of $\Theta(m)$ by a partition menu. We next show that the revenue by selling the grand bundle is much lower. Specifically, for any menu of the principal in which the only set offered is the grand bundle at price $p$, it holds that the revenue of the principal is $O(1)$ in every equilibrium of $G_{p,\dist}$.
{We separate the price range of the principal seller into three categories: when $p < 3 \cdot H_1$, when $p \geq 3 \cdot H_{n}$, and when there exists some $i \in [n-1]$, where $3 \cdot H_i \leq p < 3 \cdot H_{i+1}.$ We will show that in the first case, the principal is selling at a constant price, resulting in $O(1)$ revenue. In the second case, the principal's price is too high and will never sell. The third case (where the grand bundle price $p$ falls between the high support of set $T_i$ and $T_{i+1}$ for index $i$) is the most interesting. 
Observe that simple arguing that the revenue is bounded by $p$ times the probability that the value is at least $p$ is not enough (e.g., for $p=H_j$ that product is at least $H_j\cdot V_j = K^{2 \cdot 3^{(j-1)}}\cdot K^{-3^{(j-1)}} = K^{3^{(j-1)}}$, which might be much larger than $m$ for large $j$). 
We will show that in this case, the higher value item sellers are incentivized to reduce their price to a point where the buyer purchases the grand bundle only if they have positive value for at least two items in $\bigcup_{j=i+1}^n T_j$. Thus the principal is in a similar situation as if they are merely selling the set $T_{i+1}$, and obtains only $O(1)$ revenue. }

We will start by considering the expected revenue of a particular item seller $j \in \bigcup_{k=i+1}^n T_k$. 

\begin{restatable}{proposition} {claimCopySellerProbabilitySale}\label{claim:copySellerProbabilitySale}
    Assume that $3 \cdot H_i \leq p < 3 \cdot H_{i+1}$ for some $i\in [n-1]$. For any item $j$ in $\bigcup_{k=i+1}^n T_k$ and any $q_j$ in the support of $s_j$, it holds that $\Prob[p - q_j \geq \sum_{k \neq i} V_k(s_k)] > \paren*{1 - 3 \paren*{x_{i+1} + \frac{1}{p \cdot x_i^2}}}$. 
\end{restatable}

The above claim can be viewed as saying conditioned on $s_j = q_j$ and $v_j \geq q_j$, the probability of the buyer not purchasing from the principal must be at least $1 - 3 (x_{i+1} + \frac{3 x_i^2}{p})$. Notice that given $p \geq 3 \cdot H_i$, in order for the buyer to be interested in the bundle, their value for \textit{some} item in $\bigcup_{k=i+1}^n T_k$ must be high. We can now apply a union bound to bound the probability of the buyer purchasing the grand bundle, and as a result, the principal seller's revenue. 

\begin{restatable}{lemma}{lemPrincipalRevPricingAtSI} \label{lem:PrincipalRevPricingAtSI}
    For any grand-bundle menu $p$ such that $3 \cdot H_n> p\geq 3\cdot K^2$, 
    the revenue of the principal in any NE of the game $G_{p,\dist}$ is at most $36$.
\end{restatable}

{Using the above two claims we bound the revenue of any grand-bundle pricing in any equilibrium by a constant.}
\begin{restatable}{lemma}{lemBrevBound}\label{lem:brevBound}
    For any grand-bundle menu $p$, the revenue of the principal in any Nash equilibrium of the game $G_{p,\dist}$ is $O(1)$. 
\end{restatable}

\section{Conclusion and Future Work}

Our work explores revenue guarantee of the principal seller in a oligopoly setting, particularly when facing agile item competitors who set the prices of their items in response to the principal seller's pricing strategy. We study the ability of the principal to extract revenue by using her unique power to price bundles of items.

Our work is by no means comprehensive in this space and we hope we can inspire future work on auction pricing with competitors.  First, on a technical level, we expect there is ample room to 
generalize or extend 
our approximation result.  This could involve establishing other approximation results under different sufficient conditions on the value distributions, developing additional upper bounds on the principal seller's revenue, and/or exploring the power of more complex bundling strategies.

Another direction for follow-up work is to expand the set of allowable buyer and/or seller preferences.  For example, one could investigate the impact of non-additive buyer valuations, such as submodular or subadditive valuations.  One could also introduce production costs for the sellers, which might impact pricing and/or entry decisions by the potential competitors.

Yet another direction is to relax our assumptions on the nature of the competition faced by the principal seller.  For example, in our model the item sellers set prices independently at equilibrium.  However, it is also natural to consider item sellers who  collude with each other in the way they set prices, either explicitly or via algorithmic collusion due to using similar off-the-shelf pricing algorithms.  One could consider separately the cases where colluding sellers can transfer utility between each other or not.  Taking this one step further, one could also consider competitors who own multiple items and can offer bundles for the items available to them; this can be thought of as a form of advanced collusion with both transfers and the ability to bundle goods.

Finally, our work focuses on revenue maximization from the principal seller's perspective. Another related angle of inspection is the buyer's surplus: namely, how much does the buyer's surplus improve, relative to the monopolist case, if the principal seller has item competitors for some or all items? This could have implications for policy considerations such as antitrust concerns, where one might seek to understand the extent to which item sellers serve as effective competition for a dominant principal seller.

\bibliography{ref}
\bibliographystyle{alpha}
\appendix
\section{Illustration for Price Sensitivity.} \label{app:sec:price-sensitivity}
\begin{figure}[htbp]
\centering
\includegraphics[width=0.8\textwidth]{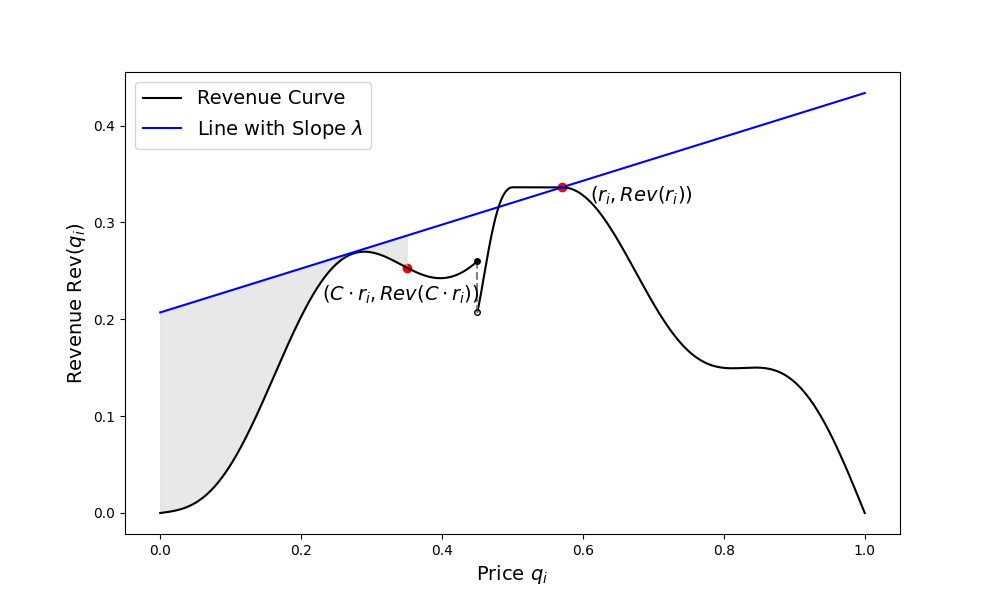}
\caption{$(\lambda, C)$-Price Sensitive Revenue Curve. For distribution $\dist_i$, the revenue function $\rev(q_i)= q_i\cdot (1-\dist_i(q_i))$, must be below the line with slope $\lambda$ that passes through $(r_i, \rev(r_i))$ for every price below $C\cdot r_i$, where $r_i$ is the maximal Myerson price for $\dist_i$.} 
\label{fig:price-sensitivity} 
\end{figure}
\section{Omitted Proofs in Section \ref{sec:upperbound}} \label{app:sec:upperbound}

\claimValueChange*
\begin{proof}
    The set $T$ maximizes $\diff_{\vec{v}}(T,p,\vec{q})$ and the set $T'$ maximizes $\diff_{\vec{v}'}(T',p,\vec{q})$.  If $T \neq T'$ then we must have $\diff_{\vec{v}}(T,p,\vec{q}) \geq \diff_{\vec{v}}(T',p,\vec{q})$ and $\diff_{\vec{v}'}(T',p,\vec{q}) \geq \diff_{\vec{v}'}(T,p,\vec{q})$, with at least one of them strict (since we break ties consistently).  So in particular we must have
    \begin{align}
    \label{eq:diff.ineq} 
    \diff_{\vec{v}}(T,p,\vec{q}) + \diff_{\vec{v}'}(T',p,\vec{q}) > \diff_{\vec{v}}(T',p,\vec{q}) + \diff_{\vec{v}'}(T,p,\vec{q}).
    \end{align}
    
    On the other hand, for any $T$ and $T'$ with $T \cap S = T' \cap S$, 
    we must have $\diff_{\vec{v}}(T,p,\vec{q}) - \diff_{\vec{v}'}(T,p,\vec{q}) = \sum_{i \in T \cap S}(v_i - v'_i) = \sum_{i \in T' \cap S}(v_i - v'_i) = \diff_{\vec{v}}(T',p,\vec{q}) - \diff_{\vec{v}'}(T',p,\vec{q})$.  As this is inconsistent with \eqref{eq:diff.ineq}, we conclude that $T = T'$ as required.
\end{proof}

\claimSameMaxMenu*
\begin{proof}
   Fix $q_{-i}$, given $q_i < q'_i$, then for any $T_i$ such that $i \in T_i$, the utility difference for the buyer for purchasing from the item seller side items in $T_i$ is just the difference of utility for the particular item.  That is,\footnote{Here and elsewhere, we will use $(x)_+$ to mean $\max\{x,0\}$.}
    $$\ub_{I, \vec{v}}(T_i, (q_i, q_{-i})) - \ub_{I, \vec{v}}(T_i, (q'_i, q_{-i})) = (v_i - q_i)_+ - (v_i - q'_i)_+ = (\min(v_i, q'_i) - q_i)_+.$$ 
   On the other hand, {for any $T_{\neg i}$ such that $i \not \in {{T_{\neg i}}}$,} there is no utility difference, namely  $$\ub_{I, \vec{v}}({{T_{\neg i}}}, (q_i, q_{-i})) - \ub_{I, \vec{v}}({{T_{\neg i}}}, (q'_i, q_{-i})) = 0.$$ 

    {As $\diff_{\vec{v}}(T, p, \vec{q}) = \ub_{P, \vec{v}}(T, p) - \ub_{I, \vec{v}}(T, \vec{q})$
    and $\diff_{\vec{v}}(T, p, (q'_i, q_{-i})) = \ub_{P, \vec{v}}(T, p) - \ub_{I, \vec{v}}(T, (q'_i, q_{-i}))$, we conclude that} 
    \begin{align}
    \label{eq:diff.i.in.T}
    \diff_{\vec{v}}(T, p, (q'_i, q_{-i})) = \diff_{\vec{v}}(T, p, (q_i, q_{-i})) + (\min(v_i, q'_i) - q_i)_+ 
    \quad\text{ for any } T \ni i
    \end{align}
    and 
    \begin{align}
    \label{eq:diff.i.not.in.T}
    \diff_{\vec{v}}({{T_{\neg i}}}, p, (q_i, q_{-i})) = \diff_{\vec{v}}({{T_{\neg i}}}, p, (q'_i, q_{-i})) \quad\text{ for any } {{T_{\neg i}}} \not \ni i.
    \end{align}  

{
    Write $T(q_i)$ for the set that the buyer purchases from the principal given $q_i$.
    Note that $T(q_i) \in \argmax_T\{ \diff_{\vec{v}}(T, p, (q_i, q_{-i})) \}$, with ties broken as per our tie-breaking rule for the buyer.  Note also that, due to our tie-breaking in favor of purchasing more items from item sellers, 
    we must have $i \in T(0)$ (since if $q_i = 0$ it is never utility-decreasing for the buyer to acquire item $i$ from item seller $i$, regardless of what was purchased from the principal).  Define $T_i = T(0)$, so that $\diff_{\vec{v}}(S, p, (0, q_{-i})) \leq \diff_{\vec{v}}(T_i, p, (0, q_{-i}))$ for all sets $S$.  Then for any $S \ni i$ and any $q_i$, equation \eqref{eq:diff.i.in.T} implies
    \begin{align*}
    \diff_{\vec{v}}(S, p, (q_i, q_{-i})) &= \diff_{\vec{v}}(S, p, (0, q_{-i})) + v_i - q_i\\ 
    &\leq \diff_{\vec{v}}(T_i, p, (0, q_{-i})) + v_i - q_i\\ 
    &= \diff_{\vec{v}}(T_i, p, (q_i, q_{-i}))
    \end{align*} 
    and hence (by consistent tie-breaking) if $S = T(q_i)$ then we must have $S = T_i$.  By the same argument, if there exists some $q_i > q'_i > 0$ and sets $S,T \not\ni i$ such that $T = T(q_i)$ and $S = T(q'_i)$, then we must have $S = T$.  We can therefore write $T_{\neg i}$ for the unique set $T_{\neg i} \not\ni i$ purchased at any choice of $q_i$.  Finally, since $\diff_{\vec{v}}(T_{\neg i}, p, (q_i, q_{-i}))$ is independent of $q_i$ whereas $\diff_{\vec{v}}(T_i, p, (q_i, q_{-i}))$ is decreasing, there exists a threshold $\theta_i$ such that $T_i$ is chosen for $q_i \leq \theta_i$ and $T_{\neg i}$ is chosen for $q_i > \theta_i$.
    }
\end{proof}

\claimProbMonotone*
\begin{proof}
    By \Cref{claim:sameMaxMenu} 
    there exist sets $T_i \ni i$ and $T_{\neg i} \not\ni i$ and a threshold $\theta_i \geq 0$ such that the set that the buyer purchases from the principal is $T_i$ if $q_i \leq \theta_i$, and is $T_{\neg i}$ if $q_i > \theta_i$. 
    For any $q_i > 0$, it is utility-maximizing for the buyer to purchase item $i$ from item seller $i$ precisely if (a) item $i$ is not purchased from the principal and (b) $v_i \geq q_i$.  We conclude that $\ind_{i,\vec{v}}(p, \vec{q}) = 0$ for all $q_i > \max\{v_i, \theta_i\}$, and $\ind_{i,\vec{v}}(p, \vec{q}) = 1$ for all $0 < q_i < \max\{v_i, \theta_i\}$.  For the remaining case that $q_i = 0$, we note that the buyer will always purchase from the item seller regardless of what was purchased from the principal.  We conclude that  $\ind_{i,\vec{v}}(p, (q_i, q_{-i}))$ is monotone non-decreasing in $q_i$, as claimed.
\end{proof}

\lemSubsetBound*
\begin{proof}
    Write $Rev(p)$ for the expected revenue under $(p,\vec{s})$.  What we will show is that $Rev(Y) \geq Rev(p) - Wel(Z)$.  Indeed, consider the following (direct revelation) mechanism for the principal to sell the items in $Y$ only.  The principal solicits reported values $v_i \sim F_i$ for $i \in Y$, then draws $\tilde{v}_i \sim \dist'_i$ for $i \in Z$, then executes the mechanism for $M$ described by price menu $p$.  If any item $i \in Z$ is to be allocated under this mechanism, we instead transfer the simulated value $\tilde{v}_i$ to the buyer.  This randomized mechanism has expected revenue at least $Rev(p) - Wel(Z)$, as $Wel(Z)$ is the maximum expected transfer to the buyer.  Furthermore, as the only randomization is in the simulation of $\tilde{v}_i$, we can derandomize by choosing the revenue-maximizing choice of $\tilde{v}$ (since the mechanism is dominant strategy incentive compatible, as there is only a single buyer).  The resulting deterministic mechanism can be described as a deterministic menu, so its expected revenue must be at most $Rev(Y)$.  We conclude that $Rev(Y) \geq Rev(p) - Wel(Z)$, as claimed.
\end{proof}
\section{Omitted Proofs in Section \ref{sec:approxopt}} \label{app:sec:approxopt}

\subsection{Omitted Proofs in Section \ref{subsec:approxopt:implications} } \label{app:subsec:approx-implication}

\lemPriceSense*
\begin{proof}

By the definition of $(\lambda, C)$-price sensitivity, $G$ is $(\lambda, C)$-price sensitive if and only if for $C$ it holds that $\forall C' \in [0, C]$, $\lambda \leq \frac{\rev_{G}(r_G) - \rev_{G}(C' \cdot r_G)}{(1-C') \cdot r_G}$.  Clearly $\lambda = \min_{C' \leq C} \frac{\rev_{G}(r_G) - \rev_{G}(C' \cdot r_G)}{(1-C') \cdot r_G}$ satisfies this condition. 
\end{proof}

\lemIIDVar*
\begin{proof}
    Since the items are i.i.d, we know that the standard deviation of the buyer's truncated value for $m$ items is $\sqrt{m}$ times the standard deviation of the buyer's truncated value for one item. Hence $\sigma(V(\vec{r})) = \sqrt{m} \cdot \sigma_{G}(r_G)$. When $m \geq \frac{144}{\lambda^3(1-C)^3} \cdot \paren*{\frac{r_G}{\sigma_G(r_G)}}^2$, $\sqrt{m} \geq \frac{12}{(\lambda(1-C))^{3/2}} \cdot \frac{r_G}{\sigma_G(r_G)}$, and hence $\sigma(V(\vec{r})) = \sqrt{m} \cdot \sigma_{G}(r_G) \geq \frac{12}{(\lambda(1-C))^{3/2}} \cdot r_G$. 
\end{proof}

\corIIDRev*
\begin{proof}
 Since $r_G$ is a unique Myerson price, by  \Cref{lem:price-sense}, $G$ is $(\lambda, C)$-price sensitive for $C = 1 - \frac{G(r_G)^4}{9}$ and some $\lambda(G) > 0$ (otherwise by the definition of $\lambda(G)$, there must exists a point on the revenue curve which has the same value as $\rev_{G}(r_G)$, violating the uniqueness assumption). Since $r_G > \inf(G)$, then $1- C = \frac{G(r_G)^4}{9} > 0$. Additionally, $r_G > \inf(G)$  ensures that $\sigma_{G}(r_G) > 0$, since there is some variation in the values. By \Cref{lem:iid-var}, for $m$ that is polynomial in  $1/G(r_G)$, $1/\lambda(G)$ and $\sigma_{G}(r_G)$, the standard deviation requirement in \Cref{thm:generalPrincipalBundleRev} is satisfied, and hence the theorem result applies.  
\end{proof}

\lemSmoothPriceSense*
\begin{proof}
    Since $G$ is $\delta$-revenue-concave,
    we know that for the revenue function $h(x) = x \cdot (1 - G(x))$, $h''(x) \leq - \delta/r_G$ by definition.  
    By the definition of $r_G$, $h'(r_G) = 0
    $.  
    Hence for any $q < r_G$, 
    \begin{align*}
        h'(q) \geq h'(q) - h'(r_G) = - \int_{q}^{r_G}  h''(x)\, dx \geq  (r_G - q) \cdot \frac{\delta}{r_{G}},
    \end{align*}
    and thus
    \begin{align*}
         h(r_G) - h(C\cdot r_G) &= \int_{C \cdot r_G}^{r_G} h'(x) \geq \int_{C \cdot r_G}^{r_G} (r_G - x) \cdot \frac{\delta}{r_{G}} \, dx = \frac{\delta}{r_{G}} \cdot \left( r_G \cdot x - \frac{1}{2} x^2 \right)\biggr|^{r_G}_{C \cdot r_G} = \frac{1}{2} \delta (1-C)^2 r_G.
    \end{align*} 

    Now, by \Cref{lem:price-sense}$, \lambda = \min_{C' \leq C} \frac{\rev_{G}(r_G) - \rev_{G}(C' \cdot r_G)}{(1-C') \cdot r_G}$. Since $G$ is $\delta$-revenue-concave, $\frac{\rev_{G}(r_G) - \rev_{G}(C' \cdot r_G)}{(1-C') \cdot r_G}$ is minimized when $C' = C$. Hence 
    \begin{align*}
        \lambda = \frac{\rev_{G}(r_G) - \rev_{G}(C \cdot r_G)}{(1-C) \cdot r_G} = \frac{\frac{1}{2} \delta (1-C)^2 r_G}{(1-C) \cdot r_G} = \frac{\delta}{2} \cdot (1-C). 
    \end{align*}
\end{proof}

\lemSmoothVariance*
\begin{proof}
Since $G$ is $\delta$-smooth, for any $x \in (0, 1)$, $G'(x) \geq \frac{\delta}{r_{G}}$. Hence for any $x \in (0, 1)$, $G(x) = \int_{0}^{x} G'(x)\, dx \geq x \cdot \frac{\delta}{r_G}$, hence 
\begin{align*}
   r_G - \E_{v \in G}[\min(v, r_G)] = r_G - \mu_{G}(r_G) = \int_{0}^{r_G} G(x)\, dx \geq\int_{0}^{r_G}  x \cdot \frac{\delta}{r_G}\, dx \geq \frac{\delta \cdot r_G}{2}. 
\end{align*}

Now, we will bound the variance of $V_{G}(r_G) = \min(v, r_G)$.  
\begin{align*}
   var(V_{G}(r_G)) &= \E[(V_{G}(r_G) - \mu_G(r_G))^2] \geq \int_{0}^{r_G} G'(x) \cdot (x - \mu_G(r_G))^2\, dx\\
   &\geq \int_{0}^{r_G} \frac{\delta}{r_G} \cdot (x - \mu_G(r_G))^2\, dx = \frac{\delta}{r_G} \cdot \frac{1}{3} \left(x - \mu_G(r_G)\right)^3 \biggr|_{0}^{r_G}\\
   &= \frac{\delta}{3r_G} \left(\cdot (r_G - \mu_G(r_G))^3 + \mu_G(r_G)^3 \right) \\
   &\geq \frac{\delta}{3r_G} \cdot \frac{(r_G)^3}{4} = \frac{\delta \cdot (r_G)^2}{12}. 
\end{align*}
The first to second line is by the bound on $G'(x)$, the second to third line is just by evaluating the integral, and the third to forth line is by minimizing the convex function $(r_G - x)^3 + x^3$. 
\end{proof}

\corSmoothRev*
\begin{proof}
   Observe that $r_G - \E_{v \sim G}[\min(v, r_G)] = \Omega(\delta) \cdot r_G$ and that 
   $r_{max} \geq r_G \geq r_{min}$. This imply that for any $\delta$-smooth distributions $\dist_1,\dist_2,\ldots,\dist_m$,
   for all $i$, it holds that   
$K = K(\dist_1,\dist_2,\ldots,\dist_m) \leq \frac{r_{max}}{\Omega(\delta) \cdot r_{min}}$. Note that by $\delta$-smoothness, $F_i(r_i) \geq \delta$. Thus for the choice of $C$ in \Cref{thm:generalPrincipalBundleRev}, $1- C = \min_{i \in \items} \frac{F_i^4(r_i)}{8K+1}= \Omega(\delta^5 \cdot \frac{r_{min}}{r_{max}})$. 
Hence by \Cref{lem:smooth-price-sense}, every distribution $\dist_i$ is $(\lambda, C)$-price sensitive for $\lambda = \frac{\delta}{2} (1-C) = \Omega(\delta^{6} \cdot \frac{r_{min}}{r_{max}})$. 

By \Cref{lem:smooth-variance}, the variance of each distribution $\dist_i$ is $\Omega(\delta \cdot r_{min}^2)$. In order to satisfy the standard deviation requirement in \Cref{thm:generalPrincipalBundleRev}, the variance for $m$ items must be at least 
\begin{align*}
    \frac{144 \cdot r^2_{max}}{\lambda^3(1-C)^3} = \frac{144 \cdot r^2_{max}}{\Omega(\delta^{6} \cdot \frac{r_{min}}{r_{max}})^3 \cdot \Omega(\delta^5 \cdot \frac{r_{min}}{r_{max}})} = O\left(
    \left(\frac{1}{\delta}\right)^{23} \cdot \left(\frac{r_{max}}{r_{min}}\right)^{4} \cdot r^2_{max}
    \right).
\end{align*}
Hence for some $m = \Omega\left(\left(\frac{1}{\delta}\right)^{24} \cdot \left(\frac{r_{max}}{r_{min}}\right)^{6}\right) = \mathsf{poly}\left(\frac{1}{\delta}, \frac{r_{max}}{r_{min}}\right)$, the variance of $m$ items satisfy the requirements in \Cref{thm:generalPrincipalBundleRev}. We can now apply \Cref{thm:generalPrincipalBundleRev} to get the bound on principal revenue. 
\end{proof}

\subsection{Omitted Proofs in Section \ref{subsec:approxopt:proof}}
For the reader's convenience, we reiterate the notations used in this section. 

For $v_i\sim \dist_i$ we use $\truncV{i}{r_i}$ to denote the random variable $\min(v_i, r_i)$: the value of item $i$ truncated at $r_i$. Note that the truncated social welfare is the random variable $V(\vec{r}) = \sum_i \truncV{i}{r_i}= \sum_i \min(v_i, r_i)$ where $v_i\sim \dist_i$ for every $i$ and $\vec{r}=(r_1,...,r_m)$. 
    The expected truncated social welfare is simple the mean on $V(\vec{r})$, and we write $\sigma(V(\vec{r}))$ to denote the standard deviation of the truncated social welfare $V(\vec{r})$. 
    We use $\mu_i(r_i) = \mu(\truncV{i}{r_i})$ to denote the mean of $\truncV{i}{r_i}$, and $\sigma_i(r_i) = \sigma(\truncV{i}{r_i})$ to denote the standard deviation of  $\truncV{i}{r_i}$.

    Expanding on the notation of $\truncV{i}{r_i}$ that we defined earlier in the section, we will use $V_i(s_i)$ to denote the random variable $\min(v_i, q_i)$, where $q_i$ is a random pricing strategy that is drawn from the item seller $i$'s mixed strategy $s_i$. We will denote $\mu_i(s_i) = \mu(V_i(s_i))$ and $\sigma_i(s_i) = \sigma(V_i(s_i))$ as the mean and standard deviation of $V_i(s_i)$, respectively. Similarly, we will use $\mu(\vec{s})$ and $\sigma(\vec{s})$. to represent the mean and standard deviation of $\sum_{j \in \items} V_j(s_j)$, respectively. We also use $\mu_{-i}(s_{-i})$ and $\sigma_{-i}(s_{-i})$ to represent the mean and standard deviation of $\sum_{j \neq i} V_j(s_j)$, respectively. For the sake of simplicity in notation, when item seller $i$'s strategy $s_i$ is pure and pricing at $q_i$, we would just use $q_i$ to represent $s_i$. For example, $V_i(C \cdot r_i)$ denote the random variable of $\min(v_i, C \cdot r_i)$ when $v_i\sim F_i$.

    For distributions $\dist_1,\ldots, \dist_m$ we define $K(\dist_1,\ldots, \dist_m) = \frac{\max_{i \in \items} \set*{r_i - \mu_i(r_i)}}{\min_{i \in \items} \set*{r_i - \mu_i(r_i)}}$ to be the ratio between the largest and smallest value of $\set*{r_i - \mu_i(r_i)}$. Also, for notational convenience, we will sometimes use $p$ to denote the principal's price of the grand bundle, instead of $p(\items)$. 
\begin{claim} \label{claim:boundVariance}
    For any $1 \geq  C > 0$ and item seller $i$'s strategy $s_i$ such that $\inf(s_i) \geq C \cdot r_i$, it holds that $\sigma_i^2(s_i) \geq \paren*{1 - \frac{2(1- C)}{C \cdot \dist_i(r_i)^4}}\cdot\sigma_i^2(r_i)$. 
\end{claim}
\begin{proof}     
        By the definition of variance,  $\sigma_{i}(r_i)^2 = \E[\truncV{i}{r_i}^2] - \mu_i(r_i)^2$ and $\sigma_{i}(s_i)^2 = \E[\truncV{i}{s_i}^2] - \mu_i(s_i)^2$. 
        We can assume WLOG that $s_i$ is supported on $[0, r_i]$, because if it is support on values that are larger than $r_i$, then truncating these values at $r_i$ will only make the variance the $\truncV{i}{s_i}$ smaller.
        
        Since all support of $s_i$ is at most $r_i$, $\E[\truncV{i}{s_i}] \leq \E[\truncV{i}{r_i}]$ and thus $\E[\truncV{i}{s_i}]^2 \leq \E[\truncV{i}{r_i}]^2$. Therefore 
        \begin{align*}
            \sigma_{i}^2(r_i) - \sigma_i^2(s_i) &= (\E[\truncV{i}{r_i}^2] - \E[\truncV{i}{r_i}]^2) - (\E[\truncV{i}{s_i}^2] - \E[\truncV{i}{s_i}]^2) \\
            &\leq \E[\truncV{i}{r_i}^2] - \E[\truncV{i}{s_i}^2] \leq (\Prange{i}{q_i}{r_i}+ \Prange{i}{r_i}{\infty}) \cdot (r_i^2 - q_i^2).
        \end{align*}
        The rest of our  proof will bound $\E[\truncV{i}{r_i}^2] - \E[\truncV{i}{s_i}^2]$ from above, as well as $\sigma_{i}^2(r_i)$ from below. This will help us establish a relationship between $\sigma_i^2(s_i)$ and $\sigma_{i}^2(r_i)$.

        Let $q_i = \inf(s_i)$. 
        We can write out $\E[\truncV{i}{r_i}^2]$ as the weighted sum over conditional probabilities using the  law of total probability over events $0 \leq \truncV{i}{r_i} \leq q_i$, $q_i \leq \truncV{i}{r_i} \leq r_i$ and $r_i \leq \truncV{i}{r_i}$. Note that when $\truncV{i}{r_i} < r_i$, $\truncV{i}{r_i} = \min(v_i, r_i)$ is really just equal to $v_i$. Meanwhile, when $\truncV{i}{r_i} \geq r_i$, $\truncV{i}{r_i} = r_i$. Hence we have 
        \begin{align}
            \E[\truncV{i}{r_i}^2] &= \Prange{i}{0}{q_i}\cdot \E[v_i^2 | v_i \leq q_i] \nonumber   + \Prange{i}{q_i}{r_i} \cdot \E[v_i^2 | q_i \leq v_i \leq r_i] + \Prange{i}{r_i}{\infty} \cdot r_i^2 \nonumber\\
            &\leq \Prange{i}{0}{q_i}\cdot \E[v_i^2 | v_i \leq q_i] +  \Prange{i}{q_i}{\infty} \cdot r_i^2
        \label{eq:risquare-expect}
        \end{align}
        Since $q_i = \inf(s_i)$, we know that $\E[\truncV{i}{s_i}^2] \geq \E[\truncV{i}{q_i}^2]$. Using the law of total probability again, we get 
        \begin{align}
            \E[\truncV{i}{s_i}^2] \geq \E[\truncV{i}{q_i}^2] &= \Prange{i}{0}{q_i} \cdot \E[v_i^2 | v_i \leq q_i] +\Prange{i}{q_i}{\infty} \cdot q_i^2 \label{eq:xisquare-variance}
        \end{align}
        By subtracting \Cref{eq:xisquare-variance} from \Cref{eq:risquare-expect}, we get: 
        \begin{align*}
            \Rightarrow \E[\truncV{i}{r_i}^2] &- \E[\truncV{i}{s_i}^2] \leq \Prange{i}{q_i}{\infty} \cdot (r_i^2 - q_i^2). 
        \end{align*}
        
        Since $r_i$ is a Myerson price, it maximizes the function $x \cdot \Prob[x \leq v_i]$. Hence for any $w$, $\Prob[v_i \geq w] \leq \frac{r_i \cdot \Prange{i}{r_i}{\infty}}{w}$. Since $q_i \geq C \cdot r_i$, it must be the case that
        \begin{align*}
            \Prange{i}{q_i}{\infty} \leq \frac{r_i \cdot \Prange{i}{r_i}{\infty}}{C \cdot r_i} = \frac{1}{C} \cdot \Prange{i}{r_i}{\infty}.
        \end{align*}
        Therefore 
        \begin{align*}
            \sigma_{i}^2(r_i) - \sigma_{i}^2(s_i) &\leq \frac{1}{C} \cdot \Prange{i}{r_i}{\infty} \cdot (r_i^2 - q_i^2) \leq \frac{1}{C} \cdot \Prange{i}{r_i}{\infty} \cdot (r_i^2 - (C\cdot r_i)^2) \\
            &=  \frac{1}{C} \cdot \Prange{i}{r_i}{\infty} \cdot (1- C) \cdot (1+C) \cdot r_i^2 
            \leq \frac{1}{C} \cdot \Prange{i}{r_i}{\infty} \cdot 2(1-C) \cdot r_i^2. 
        \end{align*}
        The last line is by $C \leq 1$. 
        
        Finally, we also bound $\sigma_i^2(r_i)$. The variance $\sigma_i^2(r_i)$ can be written as the expectation of $(\truncV{i}{r_i} -\mu_i(r_i))^2$. By the law of total expectation, 
        \begin{align*}
            \sigma_i^2(r_i) \geq \Prange{i}{r_i}{\infty} \cdot (r_i - \mu_i(r_i))^2, 
        \end{align*}
        where 
        \begin{align*}
            \mu_i(r_i) &= \int_{0}^{r_i} \Prob[v_i \geq v]\,  dv\\ 
                  &\leq \int_{0}^{r_i} \min\left(\frac{\Prange{i}{r_i}{\infty}\cdot r_i}{v}, 1\right)\,  dv\\
                  &\leq \int_{0}^{\Prange{i}{r_i}{\infty}\cdot r_i} 1\cdot\,   dv + \int_{\Prange{i}{r_i}{\infty}\cdot r_i}^{r_i} \frac{\Prange{i}{r_i}{\infty}\cdot r_i}{v}\,  dv\\
                  &\leq \Prange{i}{r_i}{\infty}\cdot r_i \cdot \left(1 + \log\left(\frac{1}{\Prange{i}{r_i}{\infty}}\right)\right). 
        \end{align*}
        For the above equation, the first to second line is by the fact that $r_i$ is the myerson price, a monopolist pricing at $v$ will get $\Prob[v \leq v_i] \cdot v \leq \Prange{i}{r_i}{\infty} \cdot r_i$, the second to fourth line are by expanding the terms and calculating the integral. Thus
        \begin{align*}
            r_i - \mu_i(r_i) \geq r_i - \Prange{i}{r_i}{\infty}\cdot r_i \cdot \paren*{1 + \log\left(\frac{1}{\Prange{i}{r_i}{\infty}}\right)} \geq 
            \left(
                1 - \Prange{i}{r_i}{\infty} \cdot \paren*{1 + \log\left(\frac{1}{\Prange{i}{r_i}{\infty}}\right)}
            \right) \cdot r_i.
        \end{align*}
        By examining the Taylor expansion of $1 - \Prange{i}{r_i}{\infty}\cdot \paren*{1 + \log\left(\frac{1}{\Prange{i}{r_i}{\infty}}\right)}$ at $\Prange{i}{r_i}{\infty}= 1$, we get that for any $\Prange{i}{r_i}{\infty}\in (0, 1]$, 
        \begin{align}
            1 - \Prange{i}{r_i}{\infty}\cdot \paren*{1 + \log\left(\frac{1}{\Prange{i}{r_i}{\infty}}\right)} \geq (1 - \Prange{i}{r_i}{\infty})^2/2.
        \end{align}
        Thus
        \begin{align*}
            \sigma_i^2(r_i)\geq \Prange{i}{r_i}{\infty}\cdot (r_i - \mu_i(r_i))^2  \geq \Prange{i}{r_i}{\infty}\cdot r_i^2 \cdot (1 - \Prange{i}{r_i}{\infty})^4/4 = \Prange{i}{r_i}{\infty}\cdot r_i^2 \cdot \dist_i(r_i)^4/4. 
        \end{align*}
 
        From both bounds on $\sigma_i^2(r_i)- \sigma_i^2(s_i)$ and $\sigma_i^2$ itself, we conclude that 
        \begin{align*}
            \frac{\sigma_i(r_i)^2 - \sigma_i^2(s_i)}{\sigma_i(r_i)^2} \leq \frac{\Prange{i}{r_i}{\infty}\cdot r_i^2 \cdot \frac{2 (1-C)}{C}}{\Prange{i}{r_i}{\infty}\cdot r_i^2 \cdot \dist_i(r_i)^4} = \frac{2 (1-C)}{C \cdot \dist_i(r_i)^4},
        \end{align*}
and thus $\sigma_i^2(s_i)\geq \paren*{1 - \frac{2(1-C)}{C \cdot \dist_i(r_i)^4}}\cdot\sigma_i^2(r_i)$. 
\end{proof}

\coroConstVarianceBound*
\begin{proof}
  Since $C \geq 1 - \frac{\dist_i(r_i)^4}{2K+1}$, then $2 (1-C) \leq 2 \cdot \frac{\dist_i(r_i)^4}{2K+1}$. Meanwhile, $C \geq 1 - \frac{\dist_i(r_i)^4}{2K+1} \geq 1 - \frac{1}{2K+1} = \frac{2K}{2K+1}$. Hence
  \begin{align*}
      \frac{2(1-C)}{C \cdot (1-\Prange{i}{r_i}{\infty})^4} \leq \frac{2 \cdot \frac{\dist_i(r_i)^4}{2K+1}}{\frac{2K}{2K+1} \cdot \dist_i(r_i)^4} = \frac{1}{K},
  \end{align*}
  and thus $1 - \frac{2(1-C)}{C \cdot \dist_i(r_i)^4} \geq 1 - \frac{1}{K}$. Plugging above inequality into \Cref{claim:boundVariance} yields our corollary.  
\end{proof}

\begin{lemma} \label{lem:meanToVar}
     Let $\inf(s_i)$ be the infinum of the support of $s_i$, then 
     \begin{align*}
         \frac{\E_{q_i \sim s_i}[\sigma_i^2(r_i) - \sigma_i^2(q_i)]}{2 (\inf(s_i) - \mu_i(\inf(s_i))}  \geq \mu_i(r_i) - \mu_i(s_i) \geq \frac{\sigma_i^2(r_i) - \sigma_i^2(s_i)}{2 (r_i - \mu_i(r_i))}
     \end{align*}
\end{lemma}
\begin{proof}
    Firstly, consider the case where $s_i$ is a pure strategy at $q_i$. Then, 
    \begin{align*}
        \mu_i(q_i) = \E[V_i(q_i)] = \int_{0}^{q_i} 1 - \dist_i(x)\, dx 
    \end{align*}
    and 
    \begin{align*}
        \sigma_i^2(q_i) = \E[V_i^2(q_i)] - \mu_i^2(q_i) &=\int_{0}^{q_i^2} 1 - \dist_i(\sqrt{x}) dx - \mu_i^2(q_i) \\
        &= \int_{0}^{q_i} (1 - \dist_i(x)) \cdot 2x \, dx - \mu_i^2(q_i). 
    \end{align*}
    Note that although the action space of the item sellers are confined to multiples of $\eps$, for the purpose of this lemma we will consider the derivative of $\mu_i(q_i)$ and $\sigma^2(q_i)$ over the domain $[0, 1]$. Our derivative computation is purely a technical tool for our proof. 
    
    Now, we can compute the derivatives of both $\mu_i(q_i)$ and $\sigma^2(q_i)$ 
    to see the rate of change of these quantities when a price reduction at $q_i$ occurs. 
    \begin{align*}
         \derivative{ \mu_i(q_i)}{ q_i} = 1 - \dist_i(q_i)
    \end{align*} and 
    \begin{align*}
        \derivative{\sigma_i^2(q_i)}{q_i} &= (1 - \dist_i(q_i)) \cdot 2q_i - 2 \mu_i(q_i) \cdot (1 - \dist_i(q_i))\\
        & = (1 - \dist_i(q_i)) \cdot (2 q_i - 2 \mu_i(q_i)).
    \end{align*}

    Let $\rem(q_i) = q_i - u_i(q_i)$, then 
    \begin{align*}
       \rem(q_i) = \int_{0}^{q_i} 1\, dx - \int_{0}^{q_i} 1 - \dist_i(x)\, dx = \int_{0}^{q_i} \dist_i(x)\, dx. 
    \end{align*}
    Clearly, $\rem(q_i)$ is monotonically non-decreasing in $q_i$. Moreover, from our computation, $\derivative{\sigma_i^2(q_i)}{q_i} = \rem(q_i) \cdot \derivative{\mu_i(q_i)}{q_i}$. Now we can use this formula to rewrite the difference between variance of $V_i(q_i)$ and $V_i(r_i)$, and bound this difference using the difference in mean.
    \begin{align*}
        \sigma^2(r_i) - \sigma^2(q_i) = \int_{q_i}^{r_i} \derivative{\sigma_i^2(x)}{q_i} \, dx 
        = \int_{q_i}^{r_i} 2 \cdot \rem(x)\cdot \derivative{\mu_i(x)}{x} \, dx.
    \end{align*}
    By the fact that $\rem(x)$ is monotonically non-decreasing, $\rem(x)$ is at most $\rem(r_i)$ and at least $\rem(q_i)$. Hence 
    \begin{align*}
        \int_{q_i}^{r_i} 2 \cdot \rem(x)\cdot \derivative{\mu_i(x)}{x} \, dx \geq 2 \rem(r_i) \int_{q_i}^{r_i}\derivative{\mu_i(x)}{x} \, dx = 2 \rem(r_i) \cdot (\mu(r_i) - \mu(q_i)). 
    \end{align*}  
    Similarly, 
    \begin{align*}
        \int_{q_i}^{r_i} 2 \cdot \rem(x)\cdot \derivative{\mu_i(x)}{x} \, dx \leq \rem(q_i) \cdot (\mu(r_i) - \mu(q_i))
    \end{align*}
    Hence for a pure strategy $q_i$, our theorem statement holds 
    \begin{align*}
        \frac{\sigma^2(r_i) - \sigma^2(q_i)}{\rem(q_i)} \geq \mu(r_i) - \mu(q_i) \geq \frac{\sigma^2(r_i) - \sigma^2(q_i)}{\rem(r_i)}. 
    \end{align*}
    For randomized strategy $s_i$, the left part of the equation holds by taking expectation over $q_i \sim s_i$: 
    \begin{align*}
       \mu_i(r_i) - \mu_i(s_i) = \E_{q_i \sim s_i}[\mu_i(r_i) - \mu_i(s_i)] \leq         \frac{\E_{q_i \sim s_i}[\sigma_i^2(r_i) - \sigma_i^2(q_i)]}{2 (\inf(s_i) - \mu_i(\inf(s_i))}. 
    \end{align*}
    For the right part of the equation, we simply use the fact that 
    \begin{align*}
        \sigma_i^2(s_i) = \E[V_i^2(s_i)] - \mu_i^2(s_i) &=\E_{q_i \sim s_i}[\E[V_i^2(q_i)]] - (\E_{q_i \sim s_i} [\mu_i(q_i)])^2 \\
        &\geq \E_{q_i \sim s_i}[\E[V_i^2(q_i)]] - \E_{q_i \sim s_i} \sq*{\mu_i(q_i)^2} = \E_{q_i \sim s_i}[\sigma_i^2(q_i)].
    \end{align*}
    Hence 
    \begin{align*}
        \mu_i(r_i) - \mu_i(s_i) \geq \frac{\E_{q_i \sim s_i}[\sigma^2(r_i) - \sigma^2(q_i)]}{\rem(r_i)} \geq \frac{\sigma^2(r_i) - \sigma_i^2(s_i)}{r_i - \mu_i(r_i)}.
    \end{align*}
\end{proof}
\begin{claim} \label{claim:remBound}
    For any $C \geq 1 - \frac{\dist_i(r_i)^4}{8K+1}$ where $K \geq 1$, it holds that $C \cdot r_i - \mu_i(C \cdot r_i) \geq \frac{4}{5} \cdot (r_i - \mu_i(r_i))$. 
\end{claim}
\begin{proof}
    Again, let $\rem(C \cdot r_i) = C \cdot r_i - \mu_i(C \cdot r_i)$ and $\rem(r_i) = r_i - \mu_i(r_i)$. Let $\rev_i(r_i) = r_i \cdot (1 - \dist_i(r_i))$. Since $r_i$ is the Myerson price, for any $x \in [0, 1]$, $x (1- \dist_i(x)) \leq \rev_i(r_i)$, hence $\dist_i(x) \geq  1 - \frac{\rev_i(r_i)}{x}$. As we have discussed in \Cref{lem:meanToVar}
    \begin{align*}
        \rem(C \cdot r_i) = \int_{0}^{C \cdot r_i} \dist_i(x)\, dx &\geq \int_{0}^{C \cdot r_i} \max\left(
            0, 1 - \frac{\rev_i(r_i)}{x}
        \right)\, dx \\
        &\geq \int_{\rev_i(r_i)}^{C \cdot r_i} 1 - \frac{\rev_i(r_i)}{x}\, dx
    \end{align*}
    When $x \geq r_i \cdot (1 - \dist_i(r_i)^4)$, 
    \begin{align*}
        1 - \frac{\rev_i(r_i)}{x} \geq 1 - \frac{r_i \cdot (1 - \dist_i(r_i))}{r_i \cdot (1 - \dist_i(r_i)^4)} \geq 1 -  \frac{1}{1 + F_i(r_i)} \geq \frac{F_i(r_i)}{2}.
    \end{align*}
    Given that $C \geq 1 - \frac{\dist_i(r_i)^4}{2K+1}$, 
    \begin{align*}
        \int_{\rev_i(r_i)}^{C \cdot r_i} 1 - \frac{\rev_i(r_i)}{x} \geq \int_{r_i \cdot (1 - \dist_i(r_i)^4)}^{C \cdot r_i} \frac{F_i(r_i)}{2} &\geq \paren*{r_i \cdot \paren*{1 - \frac{\dist_i(r_i)^4}{2K+1}} - r_i \cdot \paren*{1 - \dist_i(r_i)^4}} \cdot \frac{F_i(r_i)}{2}\\
        &= r_i \cdot \paren*{1-\frac{1}{2K+1}}  \cdot \frac{F_i(r_i)^5}{2}.
    \end{align*}
    Meanwhile, $\rem(r_i) - \rem(C \cdot r_i) = \int_{C \cdot r_i}^{r_i} \dist_i(x)\, dx$. Since it's always true that $\dist_i(x) \leq F_i(r_i)$, 
    \begin{align*}
        \rem(r_i) - \rem(C \cdot r_i) \leq (1-C) r_i \cdot F_i(r_i) \leq \frac{\dist_i(r_i)^5}{2K+1}. 
    \end{align*}
    Since $K \geq 4$, we conclude that 
    \begin{align*}
        \frac{\rem(r_i)}{\rem(C \cdot r_i)} = 1 + \frac{\rem(r_i) - \rem(C \cdot r_i)}{\rem(C \cdot r_i)} &\leq 1 + \frac{
            \frac{\dist_i(r_i)^5}{2K+1}
        }{
            \paren*{1-\frac{1}{2K+1}} \cdot \frac{F_i(r_i)^5}{2}
        }
        \\
        &\leq 1 + \frac{1/9}{(1 - 1/9) \cdot 1/2} = 1 + 1/4 = 5/4.  
    \end{align*}
    This yields our claim. 
\end{proof}

\lemHighMeanVar*
\begin{proof}
    By \Cref{lem:meanToVar}, we know that 
    \begin{align*}
        \sigma^2(\vec{r}) - \sigma^2(\vec{s}) = \sum_{i \in \items} \paren*{\sigma_i^2(r_i) - \sigma_i^2(s_i)} &\leq \sum_{i \in \items} 2 (r_i - \mu_i(r_i)) \cdot (\mu_i(r_i) - \mu_i(s_i)) \\
        &\leq 2 \max_{i \in \items}\set*{r_i - \mu_i(r_i)} \cdot \paren*{\sum_{i \in \items}\mu_i(r_i) - \sum_{i \in \items}\mu_i(s_i)}. 
    \end{align*}
    By our assumption in the lemma, $\sum_{i \in \items} \mu_i(s_i) \geq \sum_{i} \mu_i(C \cdot r_i)$. Hence 
    \begin{align*}
        \sum_{i \in \items} \mu_i(r_i) - \sum_{i \in \items}\mu_i(s_i) \leq \sum_{i \in \items} \mu_i(r_i) - \sum_{i \in \items} \mu_i(C \cdot r_i) = \sum_{i \in \items} (\mu_i(r_i) - \mu_i(C \cdot r_i)). 
    \end{align*}
    Now, again by \Cref{lem:meanToVar}, 
    \begin{align*}
        \sum_{i \in \items} (\mu_i(r_i) - \mu_i(C \cdot r_i)) &\leq \sum_{i \in \items} \frac{\sigma_i^2(r_i) - \sigma_i^2(C \cdot r_i)}{2(C \cdot r_i - \mu_i(C \cdot r_i))}\\
        & \leq \frac{1}{2(\min_{i \in \items}\set*{C \cdot r_i - \mu_i(C \cdot r_i)})} \cdot \paren*{\sum_{i \in \items} \sigma_i^2(r_i) - \sum_{i \in \items} \sigma_i^2(C \cdot r_i)}.  
    \end{align*}
    Note that our choice of $C$ satisfies the conditions in \Cref{claim:remBound}, hence for all $i \in \items$, $C \cdot r_i - \mu_i(C \cdot r_i) \geq \frac{4}{5} \cdot (r_i - \mu_i(r_i))$, and hence
    \begin{align*}
        \sum_{i \in \items} (\mu_i(r_i) - \mu_i(C \cdot r_i)) \leq \frac{5}{8(\min_{i \in \items}\set*{r_i - \mu_i(r_i)})} \cdot \paren*{\sum_{i \in \items} \sigma_i^2(r_i) - \sum_{i \in \items} \sigma_i^2(C \cdot r_i)}.
    \end{align*}
    Plugging our choice of $C$ into \Cref{coro:constVarianceBound}, we get 
    \begin{align*}
        \sum_{i \in \items} \sigma_i^2(C \cdot r_i) \geq \left(1 - \frac{1}{K}\right) \cdot \sigma_i^2(r_i), 
    \end{align*}
    hence 
    \begin{align}
        \sum_{i \in \items} \sigma_i^2(r_i) - \sum_{i \in \items} \sigma_i^2(C \cdot r_i) \leq \frac{1}{K} \cdot \sigma_i^2(r_i).
    \end{align}
    We conclude that 
    \begin{align*}
        \sum_{i \in \items} \sigma_i^2(r_i) - \sum_{i \in \items} \sigma_i^2(s_i) &\leq \frac{2 \max_{i \in \items}\set*{r_i - \mu_i(r_i)}}{8/5 \cdot \min_{i \in \items} \set*{r_i - \mu_i(r_i)}} \cdot  \frac{1}{K} \cdot \sigma_i^2(r_i) \\
        &\leq \frac{2 \max_{i \in \items}\set*{r_i - \mu_i(r_i)}}{8/5 \cdot \min_{i \in \items} \set*{r_i - \mu_i(r_i)}} \cdot \frac{1}{4 \cdot \frac{\max_{i \in \items} \set*{r_i - \mu_i(r_i)}}{\min_{i \in \items}\set*{r_i - \mu_i(r_i)}}} \cdot \sum_{i \in \items} \sigma_i^2(r_i) \leq \frac{1}{2} \cdot \sum_{i \in \items} \sigma_i^2(r_i),
    \end{align*}
    and hence 
    \begin{align*}
        \sigma^2(\vec{s}) = \sum_{i \in \items} \sigma_i^2(s_i) \geq \sum_{i \in \items} \sigma_i^2(r_i) - \frac{1}{2} \cdot \sum_{i \in \items} \sigma_i^2(r_i) = \frac{1}{2} \cdot \sum_{i \in \items} \sigma_i^2(r_i) = \sigma^2(\vec{r}). 
    \end{align*}
\end{proof}

\lemContainedResponse*
\begin{proof}
    Recall from our discussion in \Cref{sec:upperbound} (in particular, \Cref{eq:sellerUtility}), given the principal is selling the grand bundle at price $p$ and other item sellers $j$ are using strategies $s_{j}$, item seller $j$ gets the following expected revenue for pricing at $q_i$: 
    \begin{align*}
        \us_{i}(p, (q_i, s_{-i})) = q_i \cdot \Prob[\truncV{i}{r_i} \geq q_i] \cdot \Prob[p - q_i \geq \sum_{j \neq i} V_j(s_j)] = \rev_i(q_i) \cdot \Prob[p - q_i \geq \sum_{j \neq i} V_j(s_j)]. 
    \end{align*}
    Assume the buyer chooses to buy the grand bundle with probability $< 1/2$ when item seller $i$ uses strategy $s_i$. Then 
    \begin{align*}
         \Prob[p \geq \sum_{j \neq i} V_j(s_j)] \geq \Prob[p \geq \sum_{j} V_j(s_j)] \geq 1/2.
    \end{align*}
    We will show that in this case, pricing below $C \cdot r_i$ is a strictly dominated strategy for item seller $i$. In particular, we will prove that for any $\alpha \in [0, C)$, 
    \begin{align}
        \us_{i}(p, (\alpha \cdot r_i, s_{-i})) < \us_{i} (p, (r_i, s_{-i})) \label{eq:revCompare}.
    \end{align}
    Pluggin in the formula for $\us_{i}(\cdot)$, \Cref{eq:revCompare} is equivalent to 
    \begin{align*}
        \frac{\Prob[p - r_i \geq \sum_{j \neq i} V_j(s_j)]} {\Prob[p - \alpha \cdot r_i \geq \sum_{j \neq i} V_j(s_j)]} \geq \frac{\rev_i(\alpha \cdot r_i)}{\rev_i(r_i)}
    \end{align*}
    We know by all value distribution being $(\lambda, C)$-price sensitive that for any $\alpha < C$, $\rev_i(r_i) - \rev_i(\alpha \cdot r_i) \geq \lambda(1-\alpha) \cdot r_i$, and hence $\frac{\rev_i(\alpha \cdot r_i)}{\rev_i(r_i)} \leq 1 - \lambda(1-\alpha)$. Thus for our lemma, it is sufficient to prove that 
    \begin{align}
        \frac{\Prob[p - r_i \geq \sum_{j \neq i} V_j(s_j)]} {\Prob[p - \alpha \cdot r_i \geq \sum_{j \neq i} V_j(s_j)]} \geq 1 - \lambda(1-\alpha). \label{eq:itemProbRatio}
    \end{align}
    When $\Prob[p - r_i \geq \sum_{j \neq i} V_j(s_j)] \geq 1 - \lambda(1-\alpha)$, above equation is automatically satisfied. Now we will mainly focus on the case where  $\Prob[p - r_i \geq \sum_{j \neq i} V_j(s_j)] < 1 - \lambda(1-\alpha)$. In this  case, we will prove that the standard deviation of the random variables $\sum_{j \neq i} V_j(s_j)$s are large, specifically, $\sigma_{-i}(s_{-i}) = \sqrt{\sum_{j \neq i}\sigma_j^2(s_j)} \geq 2 \cdot (\frac{1}{\lambda(1-c)} + 2) \cdot \max_i r_i$. We will soon see that this large standard deviation enables us to argue that item seller $i$'s ability to affect the probability that the buyer purchases the grand bundle from the principal is limited. 
    
    Let $A = 2 \cdot (\frac{1}{\lambda(1-c)} + 2)$. Assume for contradiction that $\sigma_{-i}(s_{-i}) < A \cdot \max_i r_i$. We are considering the case that $\Prob[p - r_i \geq \sum_{j \neq i} V_j(s_j)] < 1 - \lambda(1-\alpha)$, hence 
    $\Prob[p - r_i < \sum_{j \neq i} V_j(s_j)] \geq \lambda(1-\alpha)$. Let $\Delta_i = p - r_i - \E[\sum_{j \neq i} V_j(s_j)]$.  By chebyshev's inequality, when $\Delta_i \geq 0$, 
    \begin{align*}
        \lambda(1-\alpha) &\leq \Prob \sq*{\sum_{j \neq i} V_j(s_j) > p - r_i} \\
        &=  \Prob\sq*{\sum_{j \neq i} V_j(s_j) \geq \E[\sum_{j \neq i} V_j(s_j)] + \frac{(p - r_i - \E[\sum_{j \neq i} V_j(s_j)])}{\sigma_{-i}(s_{-i})} \cdot \sigma_{-i}(s_{-i})}\\
        &\leq \frac{\sigma_{-i}^2(s_{-i})}{\Delta_i^2}.
    \end{align*}
    Hence
    \begin{align*}
       \Delta_i \leq \frac{\sigma_{-i}(s_i)}{\sqrt{\lambda(1-\alpha)}} < \frac{A \cdot \max_{j \in \items} \{r_j\}}{\sqrt{\lambda(1-\alpha)}}. 
    \end{align*}
    This means that the difference between the mean of $\sum_{j \neq i} V_j(s_j)$ and $p - r_i$ is small. Specifically,
    \begin{align*}
        \E\left[\sum_{j \neq i} V_j(s_j)\right] &\geq p - r_i - \frac{A \cdot \max_{j \in \items} \{r_j\}}{\sqrt{\lambda(1-\alpha)}} \geq p - \paren*{1 + \frac{A}{\sqrt{\lambda(1-\alpha)}}}\cdot \max_{j \in \items} \{r_j\}\\
        &\geq \sum_{i \in \items}\E[\min(v_i, C \cdot r_i)] + \frac{\frac{12}{(\lambda(1-C))^{3/2}} \cdot \max_{j \in \items} \{r_j\}}{4} - \paren*{1 + \frac{A}{\sqrt{\lambda(1-\alpha)}}}\cdot \max_{j \in \items} \{r_j\}\\
        &\geq \sum_{i \in \items}\E[\min(v_i, C \cdot r_i)].
    \end{align*}
    Now, this means that $\E\left[\sum_{j \in \items} V_j(s_j)\right] =  \E[\sum_{j \neq i} V_j(s_j)] + V_i(s_i) \geq \sum_{i \in \items}\E[\min(v_i, C \cdot r_i)]$. By \Cref{lem:highMeanVar}, we can bound the variance of  random variable $\sum_{j \neq i} V_j(s_j)$: 
    \begin{align*}
        \sigma_{-i}^2(s_{-i}) = \sum_{j \neq i} \sigma^2_j(s_j) = \sum_{j \in \items} \sigma^2_j(s_j) - \sigma^2_i(s_i) &\geq \frac{1}{2} \sum_{j \in \items} \sigma^2_j(r_j)  - r_i^2 \\
        &\geq \frac{36}{(\lambda(1-C))^3} \cdot (\max_{j \in \items} \{r_j\})^2 - r_i^2 \geq \frac{35}{(\lambda(1-C))^3} \cdot (\max_{j \in \items} \{r_j\})^2
    \end{align*}
    Hence $\sigma_{-i}(s_{-i}) \geq \frac{5}{(\lambda(1-C))^{3/2}} \cdot \max_{j \in \items} \{r_j\}$, which is a contradiction to our assumption that $\sigma_{-i}(s_{-i})< A \cdot  \max_{j \in \items} \{r_j\}$. 

    Now that we have proven $\sigma_{-i}(s_{-i}) > A \cdot \max_{j \in \items} \{r_j\}$, we will proceed to prove that this implies \Cref{eq:itemProbRatio}. 

    Let $\delta_i = 0.5606 \cdot \frac{1}{\sigma_{-i}(s_{-i})} \cdot \max_{j \neq i} \frac{\E[|V_j(s_j) - E[V_j(s_j)]|^3]}{\E[|V_j(s_j) - E[V_j(s_j)]|^2]}$. Since $V_j(s_j)$ is always at most $r_j$, $\frac{\E[|V_j(s_j) - E[V_j(s_j)]|^3]}{\E[|V_j(s_j) - E[V_j(s_j)]|^2]} \leq r_j$. Hence $\delta_i \leq 0.5606 \cdot \frac{\max_{j \in \items}\{r_j\} }{A \cdot \max_{j \in \items} \{r_j\}} \leq \frac{0.5606}{A}$. Recall that $\Phi$ denotes the CDF of the standard normal distribution $\N(0, 1)$. 
    By the two bounds on normal distribution from the Berry-Esseen Theorem (\Cref{claim:BerryEsseenGeneral}), we know that for 
    \begin{align*}
        \Prob\sq*{p - \alpha \cdot r_i \geq \sum_{j \neq i} V_j(s_j)} &= \Prob\sq*{\sum_{j\neq i } V_j(s_j) \leq  \mu_{-i}(s_{-i}) + \paren*{\frac{p - \alpha \cdot r_i - \mu_{-i}(s_{-i})}{\sigma_{-i}(s_{-i})}} \cdot  \sigma_{-i}(s_{-i})} \\
        &\leq \Phi \paren*{\frac{p - \alpha \cdot r_i - \mu_{-i}(s_{-i})}{\sigma_{-i}(s_{-i})}} + \delta_i
    \end{align*} and 
    \begin{align*}
        \Prob\sq*{p  - r_i \geq \sum_{j \neq i} V_j(s_j)} &= \Prob\sq*{\sum_{j\neq i } V_j(s_j) \leq \mu_{-i}(s_{-i}) + \paren*{\frac{p - r_i - \mu_{-i}(s_{-i})}{\sigma_{-i}(s_{-i})}} \cdot  \sigma_{-i}(s_{-i})} \\
    &\geq \Phi \paren*{\frac{p - r_i - \mu_{-i}(s_{-i})}{\sigma_{-i}(s_{-i})}} - \delta_i.
    \end{align*}

    Note that the probability density of $\N(0, 1)$ is at most $\frac{1}{\sqrt{2\pi}}$ at any point. Therefore 
    \begin{align*}
        \Phi \paren*{\frac{p - \alpha \cdot r_i - \mu_{-i}(s_{-i})}{\sigma_{-i}(s_{-i})}} - \Phi \paren*{\frac{p - r_i - \mu_{-i}(s_{-i})}{\sigma_{-i}(s_{-i})}}
        \leq  \frac{1}{\sqrt{2\pi}} \cdot (1 - \alpha) \cdot \frac{r_i}{\sigma_{-i}(s_{-i})} \leq \frac{1-\alpha}{\sqrt{2\pi} \cdot A}. 
\end{align*}

We conclude that 
\begin{align*}
    \Prob[p - \alpha \cdot r_i\geq \sum_{j \neq i} V_j(s_j)] - \Prob[p - r_i \geq \sum_{j \neq i} V_j(s_j)] &\leq \frac{1-\alpha}{\sqrt{2\pi} \cdot A} + 2\delta_i \leq \paren*{\frac{1-\alpha}{\sqrt{2\pi}} + 1.1212}\cdot \frac{1}{A}.
\end{align*}
Similarly, we know that 
\begin{align*}
    \Prob[p \geq \sum_{j \neq i} V_j(s_j)] - \Prob[p - \alpha \cdot r_i \geq \sum_{j \neq i} V_j(s_j)] &\leq \frac{\alpha}{\sqrt{2\pi} \cdot A} + 2\delta_i \leq \paren*{\frac{\alpha}{\sqrt{2\pi}} + 1.1212}\cdot \frac{1}{A}.
\end{align*}
Since we know from our assumption in the lemma that $\Prob[p \geq \sum_{j \neq i} V_j(s_j)] \geq 1/2$, then $\Prob[p - \alpha \cdot r_i \geq \sum_{j \neq i} V_j(s_j)] \geq 1/2 - \paren*{\frac{\alpha}{\sqrt{2\pi}} + 1.1212}\cdot \frac{1}{A}$. Note that proving \Cref{eq:itemProbRatio} is equivalent to proving 
\begin{align*}
    \frac{\Prob[p - \alpha \cdot r_i \geq \sum_{j \neq i} V_j(s_j)] - \Prob[p - r_i \geq \sum_{j \neq i} V_j(s_j)]} {(1-\alpha) \cdot \Prob[p - \alpha \cdot r_i \geq \sum_{j \neq i} V_j(s_j)]} \leq \lambda.
\end{align*}
We can now plug in the quantities in the left hand side of the above equation: 
\begin{align*}
    \frac{\Prob[p - \alpha \cdot r_i \geq \sum_{j \neq i} V_j(s_j)] - \Prob[p - r_i \geq \sum_{j \neq i} V_j(s_j)]} {(1-\alpha) \cdot \Prob[p - \alpha \cdot r_i \geq \sum_{j \neq i} V_j(s_j)]} &\leq \frac{\paren*{\frac{1 - \alpha}{\sqrt{2\pi}} + 1.1212}\cdot \frac{1}{A}}{(1-\alpha) \cdot \paren*{1/2 - \paren*{\frac{\alpha}{\sqrt{2\pi}} + 1.1212}\cdot \frac{1}{A}}}\\
    &\leq \frac{\paren*{\frac{1}{\sqrt{2\pi}} + \frac{1.1212}{1 - \alpha}}}{A/2 - \paren*{\frac{\alpha}{\sqrt{2\pi}} + 1.1212}}
\end{align*}
Since $\alpha \leq C$, and the above ratio is monotonically increasing in $\alpha$, 
\begin{align*}
     \frac{\paren*{\frac{1}{\sqrt{2\pi}} + \frac{1.1212}{1 - \alpha}}}{A/2 - \paren*{\frac{\alpha}{\sqrt{2\pi}} + 1.1212}} 
    \leq 
    \frac{\paren*{\frac{1}{\sqrt{2\pi}} + \frac{1.1212}{1 - C}}}{A/2 - \paren*{\frac{C}{\sqrt{2\pi}} + 1.1212}}.
\end{align*}
For $A \geq 2 \cdot (\frac{1}{\lambda(1-c)} + 2)$, one can verify that the above quantity is at most $\lambda$. This proves our lemma. 
\end{proof}

\thmGeneralPrincipalBundleRev*
\begin{proof}{(of \Cref{thm:generalPrincipalBundleRev})}
    Consider an arbitrary equilibrium $(s_1, \cdots, s_n)$ in $G_{p,\dist}$, where $p = \sum_{j}\E[V_j(C \cdot r_j)] + \frac{\sigma(V(\vec{r}))}{2}$.
    Note that as discussed in \Cref{subsec:buyerDecision}, when the buyer's value for the items is $\vec{v}$, the buyer's utility for the principal seller selling the grand bundle is $\ub_{P, \vec{v}}(T, p) = \sum_{i \in \items}v_i - p.$ Meanwhile, the buyer's maximum utility from purchasing any subset of items from item sellers who are pricing at $\vec{q}$ is $\ub_{I, \vec{v}}(\items, \vec{q}) = \max_{S \subseteq \items} \sum_{i \in S}(v_i - q_i) = \sum_{i \in \items}(v_i - q_i)_+.$ Hence the buyer only purchases from the principal when their utility from the principal is larger than that of item sellers, i.e. $\sum_{i \in \items}v_i - p > \sum_{i \in \items}(v_i - q_i)_+$. Equivalently, $p < \sum_{i \in \items} \min(v_i, q_i)$. This happens with probability $\Prob[p < \sum_{j} V_j(s_j)]$. 
    
    If $\Prob[p < \sum_{j} V_j(s_j)] \geq 1/2$, then the principal's revenue is at least $\Prob[p < \sum_{j} V_j(s_j)] \cdot p$. 
    We know that 
    \begin{align*}
        p \geq \sum_{j}\E[V_j(C \cdot r_j)] = \sum_{j}\E[\min(v_j, C \cdot r_j)] &\geq \sum_{j}\E[\min(C \cdot v_j, C \cdot r_j)] \\
        &= C \cdot  \sum_{j}\E[\min(v_j, r_j)] = C \cdot \sum_{j}\E[V_j(r_j)]. 
    \end{align*}
    
    Since $C\geq 1 - \frac{1}{9} = \frac{8}{9}$, $\frac{1}{2} \cdot C \geq \frac{1}{2}\cdot \frac{8}{9}= \frac{4}{9}> \frac{1}{3}$. Hence the principal's revenue is at least $\frac{1}{2} \cdot (C \cdot \sum_{j}\E[V_j(r_j)])> 
    \frac{1}{3}\cdot \sum_{j}\E[V_j(r_j)]$. 

    Now, if  $\Prob[p < \sum_{j} V_j(s_j)] < 1/2$, then by \Cref{lem:containedResponse}, each item seller $i$'s strategy $s_i$ is supported on $[C \cdot r_i, r_i]$. By \Cref{coro:constVarianceBound}, 
    \begin{align*}
        \sigma^2(\vec{s}) = \sum_{i \in \items} \sigma_i^2(s_i)  \geq \sum_{i \in \items} \E_{q_i \sim s_i}[\sigma_i^2(q_i)] \geq \paren*{1 - \frac{1}{K}} \cdot \sum_{i \in \items} \sigma_i^2(r_i) = \paren*{1 - \frac{1}{K}} \sigma^2(\vec{r}) \geq \frac{3}{4} \cdot \sigma^2(\vec{r}).
    \end{align*}
    Moreover, it is clear that $\mu(\vec{s}) = \E[\sum_{i} V_i(s_i)] \geq \E[\sum_{i} V_i(C \cdot r_i)] = p - \sigma(V(\vec{r}))/4 \geq p - \sqrt{\frac{4}{3}} \cdot \sigma(\vec{s})/4$. Now, the error term in the Berry-Esseen Theorem (\Cref{claim:BerryEsseenGeneral}): 
    \begin{align*}
        \delta &= 0.5606 \cdot \frac{1}{\sigma(\vec{s})} \cdot \max_{j \in \items} \frac{\E[|V_j(s_j) - E[V_j(s_j)]|^3]}{\E[|V_j(s_j) - E[V_j(s_j)]|^2]} \\
        &\leq 0.5606  \cdot \frac{1}{\sqrt{3/4} \cdot \frac{12}{(\lambda(1-C))^{3/2}} \cdot \max_{j \in \items} \{r_j\}} \cdot \max_{j \in \items} \{r_j\}
        \leq \frac{(\lambda(1-C))^{3/2}}{10}. 
    \end{align*}
    Hence by the Berry-Esseen Theorem (\Cref{claim:BerryEsseenGeneral}),  
    \begin{align*}
        \Prob\sq*{p < \sum_{j \in \items} V_j(s_j)} &= \Prob\sq*{\sum_{j \in \items} V_j(s_j) \geq  \mu(\vec{s}) + \paren*{\frac{p -\mu(\vec{s})}{\sigma(\vec{s})}} \cdot  \sigma(\vec{s})} \\
        &\geq \Phi \paren*{\frac{\mu(\vec{s}) - p}{\sigma(\vec{s})}} - \delta \geq \Phi\left(-\frac{1}{2 \sqrt{3}}\right) -  \frac{(\lambda(1-C))^{3/2}}{10} \geq 0.38, 
    \end{align*}
    Where $\Phi$ denotes the CDF of the standard normal distribution $\N(0, 1)$.
    
    Consequently, the principal seller's revenue is at least $0.38 \cdot C \cdot \E[\sum_{i \in \items} \truncV{i}{r_i}]$. By construction, $C$ is at least $1 - \frac{1}{9} = \frac{8}{9}$, hence $0.38 \cdot C \geq 0.38 \cdot \frac{8}{9} > \frac{1}{3}$. 
\end{proof}

\section{Omitted Proofs in Section \ref{sec:brev-prev-lb}} \label{app:sec:brev-prev-lb}

\propPerPartitionRev*
\begin{proof} 
    We first claim that in any equilibrium $(s_1, s_2)$  of $G_{p,\dist}$,
    the probability that $q_1 + q_2 > p$ for $q_1, q_2$ sampled from $s_1, s_2$ must be $1$.
    Consider any item seller equilibrium $(s_1, s_2)$ in $G_{p,\dist}$.
    Assume 
    that there exists $q_1, q_2$ in the support of $s_1, s_2$ such that $q_1 + q_2 \leq p$.
    Since $p = H + \varepsilon$, 
    \begin{align*}
       \min\{q_1, q_2\} \leq \frac{q_1 + q_2}{2} \leq \frac{p}{2} =  \frac{H}{2} + \frac{\varepsilon}{2}.
    \end{align*}
    Let's assume without loss of generality that item seller one has a lower price for the item ($q_1 = \min\{q_1, q_2\}$). Then, consider item seller one's revenue when pricing at $H$. Since this price is lower than the principal's price $p$, when the buyer has value $H$ for item one but value $0$ for item two, the buyer will always purchase from item seller one. This event happens with probability $x \cdot (1 - x)$ Hence 
    \begin{align*}
        \us_{1}(p, \{H, s_2\}) \geq H \cdot x \cdot (1 - x) = \frac{1}{x^2} \cdot x \cdot (1 - x) = \frac{1 - x}{x}. 
    \end{align*}
    On the other hand, when seller one price at $q_1$, their probability of sale is at most the probability that the buyer has a non zero value for the item, which is $x$. Hence 
    \begin{align*}
        \us_{1}(p, \{q_1, s_2\}) \leq q_1 \cdot x \leq  \paren*{\frac{H}{2} + \frac{\varepsilon}{2}} \cdot x \leq \frac{1}{2 x} +  \frac{\varepsilon \cdot x}{2} 
    \end{align*}
    Since $x < \frac{1}{2}$, for sufficiently small $\varepsilon$, $\frac{1}{2 x} +  \frac{\varepsilon \cdot x}{2} < \frac{1 - x}{x}$.  
    Thus, if $q_1 + q_2 \leq p$ with positive probability (over sampling $q_1\sim s_1$ and $q_2\sim s_2$) then item seller one strictly prefers deviating and pricing at $H$ over pricing at $q_1$, so $(s_1,s_2)$ cannot be an equilibrium.  

    We conclude that for any $q_1 \sim s_1$, $q_2 \sim s_2$, it must be the case that $q_1 + q_2 > p$ with probability $1$. Hence when the buyer has high value for both items (happens with probability $x^2$), the buyer would purchase from the principal. As a result, the principal revenue is at least $p \cdot x^2 \geq 1$. 
    Additionally, each seller is getting positive revenue only in the case that her item is the unique item realized to $H$. In such a case, the revenue is maximized at the Myerson price $H$, as $p>H$ (the buyer never buys from the principal in that case). We thus conclude there is a unique NE: one in which each item seller is pricing at $H$.
\end{proof}

\propPartitionEquilibrium*

\begin{proof}
    Notice that when the principal seller is using partition bundling menu with the partition $T_1$, $\cdots$, $T_n$, 
    the buyer's decision of purchase for items in each set $T_a$ is independent of information (either the buyer's value, the principal price, or item seller price) of set $T_b$ where $b \neq a$. We can use this to show that when the item sellers in $T_i$ are using identical strategies, the utility of the item sellers in $T_i$ are exactly the same between games $G_{p,\dist}$ over $\items$ and $G_{p(T_i), F_i\times F_i}$ over $T_i$. 
    
    Formally, observe that the only set $T$ which contains item $j$ and where $p(T) < \infty$ is $T_i$. Thus 
    \begin{align*}
        \us_{j}(p, \{s_j, s_{-j}\}) 
        &=
        \E_{q_j \sim s_j}\sq*{
            q_j \cdot \Prob_{v_j \sim \dist_j}[v_j \geq q_j] \cdot \Prob_{q_{-j} \sim s_{-j}, v_{-j} \sim \dist_{-j}}
            \sq*{
                \forall T \ni j, p(T) -q_j \geq \sum_{l \in T, l \neq j} \min(v_l, q_l)]
            }
        } \\
        &= \E_{q_j \sim s_j}\sq*{
            q_j \cdot \Prob_{v_j \sim \dist_j}[v_j \geq q_j] \cdot \Prob_{q_{-j} \sim s_{-j}, v_{-j} \sim \dist_{-j}}
            \sq*{
                p(T_i) -q_j \geq \sum_{l \in T_i, l \neq j} \min(v_l, q_l)]
            }
        }\\
        &= \E_{q_j \sim s_j}\sq*{
            q_j \cdot \Prob_{v_j \sim \dist_j}[v_j \geq q_j] \cdot \Prob_{q_l \sim s_l, v_{l} \sim \dist_{l}:l \in T_i, l \neq j}
            \sq*{
                p(T_i) - q_j \geq \sum_{l \in T, l \neq i} \min(v_l, q_l)]
            }
        }\\
        &= \us_{j}(p(T_i), \{s_l\}_{l \in T_i}). 
    \end{align*}

    Since the game $G_{p,\dist}$ and $G_{p(T_i), F_i\times F_i}$ are exactly the same for item sellers in $T_i$, they share the same equilibrium.
\end{proof}

\lemPrevBound*
\begin{proof}
Let the principal's pricing strategy be selling the bundle of each set $T_i$ at price $p(T_i) = H_i + \varepsilon = K^{2 \cdot 3^{(i-1)}} + \varepsilon$ for some small enough $\varepsilon$. Consider an equilibrium $(s_1, \cdots, s_n)$ in $G_{p,\dist}$. By \Cref{prop:PartitionEquilibrium}, for each $i \in [n]$, $\{s_j\}_{j \in T_i}$ are in equilibrium in the game $g_{p(T_i)}$. 
Since $H_i = \frac{1}{x_i^2}$ and each $x_i < \frac{1}{2}$, by \Cref{prop:perPartitionRev}, the principal's revenue from $g_{p(T_i)}$ is at least $1$. The principal's total revenue is then the sum over the revenue the principal gains over each set $T_i$, which must be at least $n = m/2$.      
\end{proof}

\claimCopySellerProbabilitySale*

\begin{proof}
    As we have discussed in \Cref{sec:approxopt}, when the principal is selling the grand bundle at price $p$, and when other item sellers are using pricing strategies $s_{-j}$, item seller $j$'s expected revenue when pricing at $q_j$ is equal to: 
    \begin{align*}
        \us_i(p, (q_j, s_{-j})) = \Prob\sq*{p - q_j \geq \sum_{k \neq j} V_k(s_k)} \cdot q_j \cdot \Prob[v_j \geq q_j].
    \end{align*}
     Consider the case where the item seller price at $q_j = p - 2 \cdot \sum_{l \leq i} H_l$. Thus the item seller obtains expected revenue:  
    \begin{align*}
        \us_i(p, (p - 2 \cdot \sum_{l \leq i} H_l, s_{-j})) &= \Prob\sq*{p - (p - 2 \cdot \sum_{l \leq i} H_l) \geq \sum_{k \neq j} V_k(s_k)} \cdot \paren*{p - 2 \cdot \sum_{l \leq i} H_l} \cdot \Prob[v_j \geq q_j]\\
        &\geq \Prob\sq*{2 \cdot \sum_{l \leq i} H_l \geq \sum_{k \neq j} v_k} \cdot \paren*{p - 2 \cdot \sum_{l \leq i} H_l} \cdot \Prob[v_j \geq q_j].
    \end{align*}
    (First line to second line is due to $V_k(s_k) = \min(v_k, s_k)$, and is dominated by $v_k$. ) 
    
    In our construction, for any $q_j > 0$, $\Prob[v_j \geq q_j]$ has the same value (since the buyer either has value $0$, or the high value for the item). Since $p \geq 3 \cdot H_i$, and each $H_{l}$ where $l < i$ is a vanishing fraction of $H_i$, $q_j = p - (2 + o(1)) \cdot H_i > 0$. As a result, $\Prob[v_j \geq q_j] = \Prob[v_j > 0]$. 

    When all items in set $\bigcup_{k = i+1}^n T_i$ other than item $j$ have value $0$, then the value of $\sum_{k \neq j} v_k$ is at most $2 \cdot \sum_{l \leq i} v_l$ (the inequality is tight when the buyer's value for all items in set $\bigcup_{k = 1}^i T_i$ is high). Hence  
    \begin{align*}
        \Prob\sq*{2 \cdot \sum_{l \leq i} H_l \geq \sum_{k \neq j} v_k} \geq \prod_{l=i+1}^n (1 - x_l)^2,
    \end{align*}
    and 
    \begin{align*}
        \us_i\paren*{p, \paren*{p - 2 \cdot \sum_{l \leq i} H_l, s_{-j}}} \geq \prod_{l=i+1}^n (1 - x_l)^2 \cdot  \paren*{p - 2 \cdot \sum_{l \leq i} H_l} \cdot \Prob[v_j \geq 0] > 0.  
    \end{align*}
    Since item seller $j$ clearly can make a positive profit by setting their price at $p - 2 \cdot \sum_{l \leq i} H_l$, their strategy $s_i$ would not include either pricing at $0$ or pricing at $\geq p$, either would guarantee that item seller $j$ gets $0$ revenue. Moreover, for any $q_j$ in the support of $s_j$, the item seller $j$ must get at least as much expected revenue from $q_j$, compared to pricing at $p - 2 \cdot \sum_{l \leq i} H_l$, namely, 
    \begin{align*}
        \us_i(p, (q_j, s_{-j})) = \Prob\sq*{p - q_j \geq \sum_{k \neq j} V_k(s_k)} \cdot q_j \cdot \Prob[v_j \geq q_j] \geq \prod_{l=i+1}^n (1 - x_l)^2 \cdot  \paren*{p - 2 \cdot \sum_{l \leq i} H_l} \cdot \Prob[v_j \geq 0]. 
    \end{align*}
    Since we have reasoned that $q_j > 0$, $\Prob[v_j \geq q_j] = \Prob[v_j > 0]$, and the above inequality, dividing $\Prob[v_j > 0]$ on both sides, is equivalent to 
    \begin{align*}
        \Prob\sq*{p - q_j \geq \sum_{k \neq j} V_k(s_k)} \cdot q_j \geq \prod_{l=i+1}^n (1 - x_l)^2 \cdot  \paren*{p - 2 \cdot \sum_{l \leq i} H_l}. 
    \end{align*}
    Since $q_j \leq p$ as we have reasoned, then 
    \begin{align*}
        \Prob\sq*{p - q_j \geq \sum_{k \neq j} V_k(s_k)} 
        &\geq \prod_{l=i+1}^n (1 - x_l)^2 \cdot  \frac{p - 2 \cdot \sum_{l \leq i} H_l}{p}\\
        &\geq \paren*{1 - \sum_{l=i+1}^n x_l} \cdot \paren*{1 - \frac{\sum_{l \leq i} 2 H_l}{p}}\\
        & \geq 1 - 3 x_{i+1} - \frac{3 H_i}{p} \\
        & = 1 - 3 \paren*{x_{i+1} + \frac{3}{p \cdot x_i^2}}. 
    \end{align*}
    The first line to second line is due to union bound, the second to third line is due to $H_{i}$ and $x_i$ being the dominant term in a geometric series, and the third to fourth line is because by our construction of the example, $H_i = x_i^2$. This proofs our claim. 
\end{proof}

\lemPrincipalRevPricingAtSI*

\begin{proof}
    Fix $i\in [n-1]$ such that $3 \cdot K^{2 \cdot 3^{(i-1)}}  = 3 \cdot H_i \leq  p < 3 \cdot H_{i+1} = 3 \cdot K^{2 \cdot 3^{i}}$. 
    
    The probability that the buyer purchases the grand bundle from the principal is $\Prob\sq*{p < \sum_{k \in \cup_{l=1}^n T_l} V_k(s_k)}$
    We know that when the buyer's value for all items in sets $\bigcup_{l=i+1}^n T_l$ are $0$, then 
    \begin{align*}
         \sum_{k \in \cup_{l=1}^n T_l} V_k(s_k) = \sum_{k \in \cup_{l=1}^{i} T_l} V_k(s_k) \leq\sum_{k \in \cup_{l=1}^{i} T_l} v_k = 2 \cdot \sum_{l=1}^i H_i < 3 \cdot H_i \leq p, 
    \end{align*}
    and hence the buyer would never purchase the grand bundle from the principal in this case. This means that 
    \begin{align*}
        \Prob\sq*{p < \sum_{k \in \cup_{l=1}^n T_l} V_k(s_k)} &= \Prob\cond*{
            p < \sum_{j \in \cup_{l=1}^n T_l} V_k(s_k)
        } {\sum_{k \in \cup_{l=i+1}^n T_l} v_j > 0} \\
        &\leq \sum_{j \in \cup_{l=i+1}^n T_l} \Prob\cond*{
            p < \sum_{k \in \cup_{l=1}^n T_l} V_k(s_k)
        } {v_j > 0} \cdot \Prob[v_j > 0]\\
        & = \sum_{j \in \cup_{l=i+1}^n T_l} \Prob_{q_j \sim s_j}\cond*{
            p - q_j < \sum_{k \in \cup_{l=1}^n T_l, k \neq j} V_k(s_k)
        } {v_j > 0} \cdot \Prob[v_j > 0]\\
        &\leq \sum_{j \in \cup_{l=i+1}^n T_l} 3 \paren*{x_{i+1} + \frac{1}{p \cdot x_i^2}} \cdot \Prob[v_j > 0] \\
        &= \paren*{\sum_{j \in \cup_{l=i+1}^n T_l} \Prob[v_j > 0]} \cdot 3 \paren*{x_{i+1} + \frac{1}{p \cdot x_i^2}}\\
        & = \paren*{2 \cdot \sum_{l=i+1}^n x_l} \cdot 3 \paren*{x_{i+1} + \frac{1}{p \cdot x_i^2}} \\
        &\leq 3 \cdot x_{i+1} \cdot 3 \paren*{x_{i+1} + \frac{1}{p \cdot x_i^2}} = 9 \cdot  x_{i+1} \cdot \paren*{x_{i+1} + \frac{1}{p \cdot x_i^2}}. 
    \end{align*}
    The first to second line is due to union bound; the second to third line is due to $V_j(s_j) = \min(v_j, q_j), q_j \sim s_j$ is equal to $q_j$ when $v_j \geq q_j$;  the third line to fourth line is by \Cref{claim:copySellerProbabilitySale}. The subsequent lines are simply due to exchange of terms and that the value of $x_l$ exponentially decreases with $l$. 
    
    Now, from our assumption in the lemma, we know that 
    $p < 3 \cdot H_{i+1} = 3 \cdot \frac{1}{x_{i+1}^2}$. Moreover, by our construction, $x_{i+1} = x_{i}^3$. We now conclude that the principal seller's revenue is at most 
    \begin{align*}
        p \cdot \Prob\sq*{p < \sum_{k \in \cup_{l=1}^n T_l} V_k(s_k)} \leq p \cdot 9 \cdot  x_{i+1} \cdot \paren*{x_{i+1} + \frac{1}{p \cdot x_i^2}} \\
        \leq  3 \cdot \frac{1}{x_{i+1}^2} \cdot 9 \cdot  x^2_{i+1} + \frac{9 x_{i+1}}{x_i^2} \leq 27 + 9 x_i \leq 36. 
    \end{align*}
\end{proof}

\lemBrevBound*

\begin{proof}
When $p<3\cdot H_1< 3\cdot K^2$ the revenue of the principal is bounded by $3\cdot K^2$ and thus is $O(1)$. 
When the principal is pricing at or above three times the highest value the buyer can have for an item, namely, $p \geq  3 \cdot H_n$, then the principal is pricing above the  highest value of the buyer for the grand bundle, since even if the buyer has a high value for every  item, their value for the grand bundle is 
$2 \sum_i H_i= 2 \sum_i K^{2 \cdot 3^{(i-1)}}\leq 2 \sum_i H_n/4^{i-1}< 3 \cdot H_n$. 
As a result, the principal sells with probability $0$ and obtains $0$ revenue.

In the remaining cases, there must be some $i\in [n-1]$ such that $3 \cdot K^{2 \cdot 3^{(i-1)}}  = 3 \cdot H_i \leq p<  
3 \cdot H_{i+1} = 3 \cdot K^{2 \cdot 3^{i}}$. We can then apply \Cref{lem:PrincipalRevPricingAtSI} which shows that the revenue of the principal  is at most $36$. 
This completes our proof. 
\end{proof}
\end{document}